\documentclass[10pt, journal]{IEEEtran}
\usepackage{xcolor,soul,framed} 
\colorlet{shadecolor}{yellow}
\usepackage[pdftex]{graphicx}
\graphicspath{{../pdf/}{../jpeg/}}
\DeclareGraphicsExtensions{.pdf,.jpeg,.png}
\usepackage[cmex10]{amsmath}
\usepackage{amssymb}           
\usepackage{array}
\usepackage{mdwmath}
\usepackage{mdwtab}
\usepackage{eqparbox}
\usepackage{url}
\usepackage{cite}                   
\usepackage{float}             
\usepackage{graphicx}
\usepackage{multicol}
\ifCLASSOPTIONcompsoc
\usepackage[caption=false,font=normalsize,labelfont=sf,textfont=sf]{subfig}
\else
\usepackage[caption=false,font=footnotesize]{subfig}
\fi
\usepackage{multirow}
\usepackage{algorithmicx}
\usepackage{algorithm}
\usepackage{algpseudocode}
\algnewcommand{\Initialization}[1]{
  \State \textbf{Initialization}}
\algnewcommand{\Repeatu}[1]{
  \State \textbf{Repeat}}
\algnewcommand{\until}[1]{
  \State \textbf{until}}
\usepackage{amsthm}                     
\newtheorem{theorem}{Theorem}           
\newtheorem{lemma}{Lemma}
\begin{document}
\bstctlcite{}
    \title{Unimodular Waveform Design that Minimizes PSL of Ambiguity Function over A Continuous Doppler Frequency Shift Region of Interest}
    \author{Weiting Lin,~\IEEEmembership{Student Member,~IEEE,}
            Yuwei Chang,
            Borching Su,~\IEEEmembership{Member,~IEEE}
}  
\maketitle
\begin{abstract}
In active sensing systems, waveforms with ambiguity functions (AFs) of low peak sidelobe levels (PSLs) across a time delay and Doppler frequency shift plane (delay-Doppler plane) of interest are desirable for reducing false alarms.
Additionally, unimodular waveforms are preferred due to hardware limitations.
In this paper, a new method is proposed to design unimodular waveforms with PSL suppression over a continuous Doppler frequency shift region, based on the discrete-time ambiguity function (DTAF).
Compared with existing methods that suppress PSL over grid points in the delay-Doppler plane by using the discrete ambiguity function (DAF), we regard the DTAF optimization problem as of more practical interest because the Doppler frequency shifts observed in echo signals reflected from targets are inherently continuous rather than discrete.
The problem of interest is formulated as an optimization problem with infinite constraints along with unimodular constraints.
To the best of the authors' knowledge, such a problem has not been studied yet.
We propose to reformulate a non-convex semi-infinite programming (SIP) to a semidefinite programming (SDP) with a finite number of constraints and a rank-one constraint, which is then solved by the sequential rank-one constraint relaxation (SROCR) algorithm.
Simulation results demonstrate that the proposed method outperforms existing methods in achieving a lower PSL of AF over a continuous Doppler frequency shift region of interest.
Moreover, the designed waveform can effectively prevent false alarms when detecting a target with an arbitrary velocity.
\end{abstract}
%
\IEEEpeerreviewmaketitle

\section{Introduction} \label{sec:introduction}
\IEEEPARstart{I}{n} active sensing systems, such as radar or sonar, the matched filter responses of the transmit waveforms with various time delays and Doppler frequency shifts can be evaluated using an ambiguity function (AF) \cite{Levanon04, Richards14, He12, Zhang16, Wang22}. 
An AF is considered helpful for evaluating the range-Doppler resolution of transmit waveforms \cite{Levanon04, Richards14, He12, Zhang16, Wang22}. 
In addition, low AF sidelobe levels ensure that weak targets are less likely to be masked by sidelobe levels of strong targets or distributed clutter in the matched filter output \cite{Chen22, Chitre20, Yang18}.
Moreover, waveforms with low AF sidelobe levels help to reduce the likelihood of false alarms \cite{Chen22}.

Hence, a thumbtack-like AF is desirable for its narrow central peak and near-zero levels elsewhere on the time delay and Doppler frequency shift plane (delay-Doppler plane) \cite{Levanon04, Richards14, He12, Zhang16, Wang22}.
However, suppressing AF sidelobe levels across the entire delay-Doppler plane is challenging due to the volume invariance property of AF \cite{Levanon04, Richards14, He12, Chen22, Zhang16, Wang22, Cui17, Jing19}. 
Thus, considering AF optimization within a region of interest (ROI) on the delay-Doppler plane has received attention in \cite{He12, Arlery16, Cui17, Yang18, Chitre20, Chen22, Li21, Zhang16, Jing19}.
Furthermore, waveforms with unimodular properties are preferred to keep the power amplifier (PA) in a saturated state \cite{Chen22, Yu20} due to hardware limitations such as PA efficiency requirements \cite{He12, Levanon04, Yu20}. 
To this end, unimodular waveform designs considering AF optimization within ROI on the delay-Doppler plane have been studied in \cite{He12, Arlery16, Cui17, Yang18, Chitre20, Chen22, Li21, Zhang16, Jing19}. 

In general, metrics for optimizing AF include the weighted integrated sidelobe level (WISL) \cite{He12, Arlery16, Cui17, Yang18} and the peak sidelobe level (PSL) \cite{Jing19, Chitre20, Chen22, Wu17}. 
Different methods were proposed in \cite{He12,Arlery16,Cui17} to minimize WISL of AF under unimodular constraints.
In \cite{He12}, the cyclic algorithm was utilized.
In \cite{Arlery16}, an efficient gradient method was proposed.
In \cite{Cui17}, the accelerated iterative sequential optimization (AISO) algorithm was proposed.
Additionally, in \cite{Yang18}, the AF WISL minimization problem under spectral nulling and peak-to-average ratio (PAR) constraints was tackled by an iterative sequential quartic optimization (ISQO) algorithm.
Moreover, in \cite{Zhang16}, unimodular waveforms were designed by minimizing squared errors between the desired AF template and the realizable AF.
However, tackling the WISL minimization problems in \cite{He12, Arlery16, Cui17, Yang18} or the problem in \cite{Zhang16} does not guarantee waveforms with a low PSL, even if the resulting waveform exhibits a low WISL.
Waveforms with a low PSL are of greater practical interest, as a lower AF PSL typically leads to a lower false alarm probability \cite{Sankuru21, Huang23}. 

On the other hand, unimodular waveform designs related to AF optimization problems with PSL minimization have been studied in \cite{Jing19, Chitre20, Chen22, Wu17}. 
In \cite{Jing19}, a hybrid Lagrange programming neural network (LPNN) approach was proposed to address the AF PSL minimization problem.
In \cite{Wu17}, the waveform, subject to PAR constraints, was designed to maximize the signal-to-interference-plus-noise ratio (SINR), which was equivalent to an AF optimization problem with PSL minimization and was solved using a majorization-minimization (MM) framework.
In addition, the narrowband waveform design was extended to the broadband case with AF PSL minimization in \cite{Chitre20}, where a semidefinite programming (SDP) with a rank-one constraint was tackled.
Moreover, the set of unimodular sequences was jointly designed under the scheme of mismatched filter banks considering the PSL minimization in \cite{Chen22}, and a generalized maximum block improvement (GMBI) method was proposed.

However, AF optimization problems were only considered for grid points within the ROI of the delay-Doppler plane in \cite{He12, Arlery16, Cui17, Yang18, Chitre20, Chen22, Li21, Zhang16, Jing19}, and PSL suppression was not guaranteed along a continuous Doppler frequency shift ROI. 
We suppose that AF optimization problems with PSL suppression are better considered in a continuous Doppler frequency shift ROI for the following two reasons.
Firstly, target speeds are typically not discrete values; therefore, the Doppler frequency shifts observed in echo signals reflected from targets are generally continuous rather than discrete.
Secondly, if only discrete Doppler frequency shift bins are considered, PSL suppression of AF is not guaranteed over a continuous Doppler frequency shift ROI, and sidelobes at Doppler frequencies not included in the bins of interest may be overlooked, potentially leading to false alarms.
Specifically, previous studies in \cite{He12, Arlery16, Cui17, Yang18, Chitre20, Chen22, Li21, Zhang16, Jing19} have designed unimodular waveforms based on the discrete ambiguity function (DAF) which is defined over grid points in the delay-Doppler plane.
It is worth noting that the discrete-time ambiguity function (DTAF), which accounts for a continuous Doppler frequency shift region, was defined in \cite{Zhang16}. However, in \cite{Zhang16}, unimodular waveforms were not designed directly based on the DTAF; instead, they were obtained by solving a DAF optimization problem and considered as an ``approximation" of the DTAF optimization problem.
Although a more precise ``approximation" can be achieved by discretizing the Doppler frequency shift region of interest into sufficient small intervals, the computational burden increases substantially. 
Therefore, we consider it essential to design unimodular waveforms directly based on the DTAF optimization, rather than using ``approximation" approaches as in \cite{He12, Arlery16, Cui17, Yang18, Chitre20, Chen22, Li21, Zhang16, Jing19}, particularly given the inherently continuous nature of the Doppler frequency shift.

We propose to formulate the DTAF optimization problem as a semi-infinite programming (SIP), as it includes constraints along the continuous Doppler frequency shift ROI.
To address an SIP, theorems of describing the trigonometric polynomial that is non-negative over a given segment as linear matrix inequalities (LMIs) were developed in \cite{Dumitrescu2006, Davidson2002, Dumitrescu2017},
and have been applied to several research fields, including filter designs \cite{Davidson2002, Dumitrescu2006, Dumitrescu2017}, beamformer designs \cite{Yiu03, Yu10}, grid-free compressive beamforming \cite{Xenaki15, Yang17}, and spectral estimation \cite{Dumitrescu2017}.
We utilize a theorem from \cite{Dumitrescu2006, Dumitrescu2017} to address the DTAF optimization problem with infinite constraints in the continuous Doppler frequency shift ROI. 
To the best of the authors' knowledge, considering the DTAF optimization as an SIP with PSL minimization under unimodular constraints has not yet been tackled in the literature.

In this paper, unimodular waveforms are designed by solving the DTAF optimization problem that minimizes PSL over a continuous Doppler frequency shift ROI.  
We formulate the DTAF optimization problem as an SIP that involves infinite constraints in a continuous Doppler frequency shift region. 
The problem to be tackled is non-convex owing to unimodular constraints.
We propose to transfer the non-convex SIP into an equivalent SDP with a rank-one constraint, and the resulting problem is solved using the sequential rank-one constraint relaxation (SROCR) algorithm \cite{Cao17}. 
In addition, the normalized true peak sidelobe level (NTPSL) criterion is proposed to evaluate the PSL of the DTAF across the continuous Doppler frequency shift ROI.
The simulation results demonstrate that the proposed method outperforms existing methods by a considerable margin in terms of the NTPSL criterion.
Furthermore, the waveforms designed by the proposed algorithm effectively prevent false alarms when detecting a target moving at arbitrary speeds.

The remainder of the paper is organized as follows.
In Section \ref{sec:problem_formulation}, the DTAF optimization problem with PSL minimization under unimodular constraints is formulated. 
In Section \ref{sec:proposed_algorithm}, the proposed method is introduced.
In Section \ref{sec:numerical_results}, simulation results reveal the advantages of the proposed method. 
Lastly, the conclusions are presented in Section \ref{sec:conclusion} and suggestions are provided for future research.  

\emph{Notations}:
The upper case letters in bold represent matrices, the lower case letters in bold represent column vectors, and the italic letters represent scalars, such as $\bf X$, $\bf x$, and $x$. 
The $m$-th entry of $\bf x$ is denoted by $x_m$.
The $M$-dimensional complex vector space is defined as $\mathbb{C}^{M}$.
The set of all real numbers is denoted by $\mathbb{R}$.
Zero-based indexing is applied, and $\mathbb{Z}_M$ stands for the set \mbox{$\{0,1,...,M-1\}$}. 
The set of all \mbox{$M\times M$} positive semidefinite matrices is denoted as ${\mathbb H}^M_{+}$. 
The Cartesian product of two sets $S_1$ and $S_2$ is defined as \mbox{$S_1\times S_2 = \{(x_1,x_2) \mid x_1\in S_1, x_2 \in S_2 \}$}.
The operators \mbox{$|\cdot|$}, \mbox{$(\cdot)^*$}, \mbox{$(\cdot)^T$}, and \mbox{$(\cdot)^H$} denote the modulus, the complex conjugate, the transpose, and the conjugate transpose, respectively.
The functions $\text{Tr}( {\bf X} )$ and $\text{rank}({\bf X})$ represent the trace and the rank of the matrix ${\bf X}$, respectively.
The zero matrix, the identity matrix, and the upper shift matrix of size $M\times M$ are defined, respectively, as ${\bf O}_M$, ${\bf I}_M$, 

\vspace{-3mm}
\begin{small}
\begin{equation}  \label{def:upper_shift_matrix}
\begin{aligned}
    {\bf N}_M = \left[
        \begin{matrix}
        0      & 1      & 0      & \cdots & 0      \\
        0      & 0      & 1      & \cdots & 0      \\
        \vdots & \vdots & \ddots & \ddots & \vdots \\
        \vdots & \vdots &        & \ddots & 1      \\
        0      & 0      & \cdots & \cdots & 0      \\
        \end{matrix}
        \right],
\end{aligned}
\ \text{ and } \left({\bf N}_M\right)^0 = {\bf I}_M.
\end{equation} 
\end{small}

\vspace{-2mm}
\noindent
The elementary Toeplitz matrix is defined as
\begin{equation} \label{def:Elementary_Toeplitz_matrix}
\begin{aligned}
    {\bf \Theta}_{M}^{(m)} = \left\{
        \begin{aligned}
        &\left({\bf N}_M\right)^m,   && \mbox{$m \in \{0,1,...,M-1\};$}\\
        &\left({\bf N}^T_M\right)^{|m|}, && \mbox{$m \in \{-1,-2...,-M+1\};$}\\
        &\hspace{1mm}{\bf O}_M, && \mbox{$|m|\geq M$.}
        \end{aligned}
    \right. 
\end{aligned}    
\end{equation} 
The Kronecker delta function is defined as
\begin{equation} \label{def:delta}
    {\delta}_m = 
    \left\{
    \begin{aligned}
    1, && &\mbox{if $m=0;$}\\
    0, && &\mbox{else.}\\
    \end{aligned}
    \right. 
\end{equation}  
\section{DTAF Optimization Problem with PSL Minimization under unimodular constraints}  \label{sec:problem_formulation}
The ambiguity function (AF) is considered helpful in evaluating the response of the matched filter to a signal with various time delays and Doppler frequency shifts \cite{Levanon04, Richards14, He12, Zhang16, Wang22}.
In \cite{He12, Zhang16,Jing19, Cui17, Yang18, Chen22, Wu17}, various modified definitions of the AF have been proposed based on different considerations. 
In this section, we first define the continuous-time ambiguity function (CTAF) \cite{He12}, the discrete-time ambiguity function (DTAF) \cite{Zhang16}, and the discrete ambiguity function (DAF) \cite{He12, Jing19, Cui17, Yang18, Chen22, Wu17}.
Subsequently, the DTAF optimization problem with PSL minimization under unimodular constraints will be formulated.

\subsection{Definitions of the CTAF, DTAF, and DAF}
We begin by considering the continuous signal 
\begin{equation} \label{eq:x(t)}
    x(t)=\sum_{n=0}^{N-1}x_np(t-nT_s),
\end{equation} 
where $x_n$, $\forall n \in \mathbb{Z}_N$, is the code sequence of length $N$, $p(t)$ is the pulse shaping filter, and $T_s$ is the sampling interval.
The AF of the continuous signal $x(t)$ 
is defined as \cite{He12}
\begin{equation} \label{def:CTAF}
\begin{aligned}
    A_c(\tau,f) = \int_{-\infty}^{\infty}x(t)x^*(t-\tau)e^{-j2\pi f(t-\tau)}dt,
\end{aligned}
\end{equation} 
where $\tau$ is the time delay, and $f$ is the Doppler frequency shift.
The AF defined in (\ref{def:CTAF}) is known as the CTAF \cite{Zhang16}.
By substituting (\ref{eq:x(t)}) into (\ref{def:CTAF}), after some derivations in Appendix \ref{appendix:CTAF_DTAF_deriv}, we can obtain \cite{Zhang16} 
\begin{equation} \label{eq:CTAF_DTAF}
\begin{aligned}
    \hspace{-1mm}
    A_c(\tau,f) 
    &= \sum_{l=-N+1}^{N-1}A(l,f)\chi(\tau-lT_s,f),
\end{aligned}
\end{equation}
where
\begin{equation} \label{def:A(l,f)}
\begin{aligned}
    A(l,f)= \sum_{n={\max}(0,l)}^{{\min}(N-1,N-1+l)}x_nx_{n-l}^*e^{-j2\pi fT_s(n-l)},
\end{aligned}
\end{equation}
and
\begin{equation} \label{def:chi}
\begin{aligned}
    \chi(\tau,f)= \int_{-\infty}^{\infty}p(s)p(s-\tau)e^{-j2\pi f(s-\tau)}ds.
\end{aligned}
\end{equation}
To define the DTAF, we further denote the normalized Doppler frequency shift as 
\begin{equation} \label{def:f_D}
    f_D = fT_s. 
\end{equation}
Then, by substituting \eqref{def:f_D} into \eqref{def:A(l,f)},
the DTAF is defined as \cite{Zhang16}
\begin{equation} \label{def:DTAF}
\begin{aligned}
    A(l ,f_D) = \sum_{n={\max}(0,l)}^{{\min}(N-1,N-1+l)}x_nx_{n-l}^*e^{-j2\pi f_D(n-l)},
\end{aligned}
\end{equation}
where $l\in\{0,\pm 1,\pm 2,..., \pm (N-1)\}$ and $f_D \in \left[-\tfrac{1}{2}, \tfrac{1}{2}\right]$. 
The DTAF is discrete along the time delay axis and continuous along the normalized Doppler frequency shift axis.
Moreover, from \eqref{eq:CTAF_DTAF}, it is worth noting that the shape of CTAF $A_c(\tau,f)$ can be controlled by DTAF $A(l,f_D)$ \cite{Zhang16}.
Furthermore, by sampling the normalized Doppler frequency shift region $[\tfrac{-1}{2},\tfrac{1}{2}]$ in \eqref{def:DTAF} with $2K+1$ points, $K$ is an integer, and further considering the normalized Doppler frequency shift bins 
$f_D\in \{ 0, \tfrac{\pm 1}{M}, \tfrac{\pm 2}{M},..., \tfrac{\pm K}{M} \}$, where $\tfrac{K}{M}\leq \tfrac{1}{2}$.
Then, the DAF can be defined as \cite{He12, Jing19, Cui17, Yang18, Chen22, Wu17} 
\begin{equation} \label{def:DAF}
    A_d(l ,k; M) = \sum \limits_{n={\max}(0,l)}^{{\min}(N-1,N-1+k)}x_nx_{n-l}^*e^{-j2\pi \frac{(n-l)k}{M}},
\end{equation}
where \mbox{$l\in\{0,\pm 1,..., \pm (N-1)\}$}, and \mbox{$k\in\{0,\pm 1,...,\pm K\}$}.
The DAF can be interpreted as an approximation of the DTAF \cite{Jing19}, where the accuracy of the approximation depends on the value of $M$ that determines the resolution along the normalized Doppler frequency shift axis.

\subsection{Problem formulation}
In active sensing systems, transmit waveforms with low AF sidelobe levels are desirable to prevent weak targets from being masked by the sidelobe levels of strong targets or distributed clutter in the matched filter output \cite{Chen22, Chitre20, Yang18}.
In addition, waveforms with low AF sidelobe levels help reduce the likelihood of false alarms.
Moreover, waveforms with unimodular properties are preferred to keep the PA of the transmitter in a saturated state \cite{Chen22, Yu20}, thereby meeting the PA efficiency requirements \cite{He12, Levanon04, Yu20}. 
As a result, extensive research has focused on designing unimodular waveforms with low AF sidelobe levels using criteria of the WISL \cite{He12, Arlery16, Cui17, Yang18} or the PSL \cite{Jing19, Chitre20, Chen22, Wu17}. 
However, all these works in \cite{He12, Arlery16, Cui17, Yang18, Jing19, Chitre20, Chen22, Wu17} designed unimodular waveforms based on the DAF defined in \eqref{def:DAF}, with the aim of minimizing the WISL or PSL at the grid points on the delay-Doppler plane.
We suppose that designing unimodular waveforms considering a continuous Doppler frequency shift region of interest (ROI) is more practical, as Doppler frequency shifts observed in echo signals reflected from targets are inherently continuous rather than discrete.
In other words, unimodular waveform designs based on the DTAF optimization problem are more desirable than using the DAF optimization problem.
In fact, the computational burden increases significantly as $M$ grows in an attempt to achieve a better ``approximation" of the DAF defined in \eqref{def:DAF} to the DTAF defined in \eqref{def:DTAF}.
In addition, achieving \mbox{$M \rightarrow \infty$} is impossible.
For the above reasons, we aim to design unimodular waveforms by directly tackling the DTAF optimization with PSL minimization, without resorting to the approximation using the DAF as in previous approaches \cite{He12, Arlery16, Cui17, Yang18, Jing19, Chitre20, Chen22, Wu17}.

When sidelobe levels of the DTAF are minimized within a specific region of the delay-Doppler plane, it must spread elsewhere due to the volume invariant property \cite{Zhang16} (CTAF and DAF also have this property \cite{Levanon04, Richards14, He12, Zhang16}).
As a result, we focus on the ROI within the delay-Doppler plane to suppress the PSL of the DTAF, where the ROI is defined as
\begin{equation}  \label{eq:Gamma}
    \Gamma = \left\{(l,f_D) \mid l \in \{0,\pm 1,\pm 2,..., \pm L\}, f_D \in [-f_R,f_R] \right\}
\end{equation} 
with $l$ denoting the discrete time delay bin, $L\leq N$, $N$ is the length of sequence, and $f_D$ denoting the normalized Doppler frequency shift, as defined in \eqref{def:f_D}, which lies within a continuous region \mbox{$[-f_R,f_R]\subseteq [\frac{-1}{2},\frac{1}{2}]$}.
Additionally, the main lobe of DTAF is supposed to appear at \mbox{$(l,f)=(0,0)$}.
Moreover, the zero-delay cut of the DTAF for the unimodular sequence can be approximated as \cite{Zhang16} \mbox{$|A(0, f_D)| \approx N |{\rm sinc}(\pi f_D N)|$},
where \mbox{${\rm sinc}(x)=\sin(\pi x)/\pi x$}.
Thus, the zero-delay cut of the DTAF is approximated to a sinc function which is independent of the sequence and cannot be optimized.
According to the above considerations, the sidelobe region of the DTAF is defined as
\begin{align} \label{eq:Gamma_s}
    \Gamma_s 
    &= \left\{(l,f_D) \in \Gamma \mid l \neq 0 \right\} \nonumber \\
    &= \{\pm 1,\pm 2,..., \pm L\} \times [-f_R,f_R],
\end{align}
where $\Gamma$ is defined in (\ref{eq:Gamma}).
The PSL of the DTAF is then defined as
\begin{equation} \label{def:PSL}
    \max_{(l,f_D) \in \Gamma_s}\left\{\left|A(l,f_D)\right|\right\}.
\end{equation}
The DTAF optimization problem with PSL minimization under unimodular constraints over a continuous Doppler frequency shift ROI is formulated as follows 
\begin{subequations} \label{P0}
\begin{align} 
    &\underset{\scriptstyle {\bf x} \in {\mathbb C}^N}{\mathrm{minimize}} \quad \max_{(l,f_D) \in \Gamma_s}\left|A(l,f_D)\right| \label{P0_a} \\
    &\mathrm{subject\ to} \hspace*{5mm} |x_n| = 1, \forall n \in \mathbb{Z}_N, \label{P0_b}
\end{align}
\end{subequations}
where \mbox{${\bf x} = [x_0, \ldots, x_{N-1}]^T$} is the sequence to be designed, $A(l,f_D)$ is the DTAF defined in \eqref{def:DTAF}, and $\Gamma_s$ is the sidelobe region defined in \eqref{eq:Gamma_s}.
The optimization problem \eqref{P0} is a semi-infinite programming (SIP) with infinite constraints on the normalized Doppler frequency shift ROI for each time delay bin.
Furthermore, the unimodular constraints in (\ref{P0_b}) are non-convex, making the problem even more challenging to be solved.

\section{Proposed method} \label{sec:proposed_algorithm}
In this section, the proposed method is introduced to tackle the DTAF optimization problem with PSL minimization under unimodular constraints formulated in \eqref{P0}, which is an SIP with non-convex constraints. 
We propose the following approach to reformulate the non-convex SIP with infinite constraints into a semidefinite programming (SDP) with a finite number of constraints and a rank-one constraint.   

To facilitate the following derivations,
the equivalent problem of (\ref{P0}) using an epigraph form representation is shown as follows.
\begin{subequations} \label{P2}
\begin{align} 
    &\underset{\scriptstyle t \in \mathbb{R}, \ {\bf x} \in {\mathbb C}^N}{\mathrm{minimize}} \quad t\\
    &\mathrm{subject\ to} 
    \hspace*{3mm} \left|A(l,f_D)\right|^2 \leq t, \forall (l,f_D) \in \Gamma_s \label{P2_b} \\
    &\hspace*{19mm} |x_n| = 1, \forall n \in \mathbb{Z}_N.
\end{align}
\end{subequations}

Due to the symmetry property of DTAF, $\left|A(l,f_D)\right| = \left|A(-l,-f_D)\right|$ (CTAF and DAF also has this property \cite{Levanon04,Richards14,He12}), 
it is sufficient to study positive time delay bins of DTAF as in the optimization problem \eqref{P2_2}. 
The negative time delay values of DTAF can be deduced from the symmetry property.
\begin{subequations} \label{P2_2}
\begin{align} 
    &\underset{\scriptstyle t \in \mathbb{R}, \ {\bf x} \in {\mathbb C}^N}{\mathrm{minimize}} \quad t\\
    &\mathrm{subject\ to} 
    \hspace*{0.1mm} \left|A(l,f_D)\right|^2 \leq t, \forall l \in \{1,...,L\}, \forall f_D \in [-f_R,f_R] \label{P2_2_b} \\ 
    &\hspace*{16.5mm} |x_n| = 1, \forall n \in \mathbb{Z}_N,
\end{align}
\end{subequations}
where a continuous normalized Doppler frequency shift ROI \mbox{$f_D \in [-f_R,f_R]$} is considered.

The optimization problem (\ref{P2_2}) is an SIP with infinite constraints (\ref{P2_2_b}) over a continuous normalized Doppler frequency shift ROI for each time delay bin. In the following theorem, we will show that the optimization problem (\ref{P2_2}) can be reformulated as an SDP with a finite number of constraints.
More specifically, let constants $d_0$ and $d_1$ chosen as
\begin{align}
    d_0 &= \tfrac{\tan (\pi f_R)\tan (\pi f_R)-1}{2}, \label{eq:d0} \\
    d_1 &= \tfrac{1+ \tan^2 (\pi f_R)}{4} + j\tfrac{\tan (\pi f_R)}{2}, \label{eq:d1}
\end{align}
and define matrices ${\bf{\Phi}}_{N-1}^{(n)}, n\in\mathbb{Z}_N$, as
\begin{align}
    {\bf{\Phi}}_{N-1}^{(n)} &= d_0 \boldsymbol{\Theta}_{N-1}^{(n)} + d_1^* \boldsymbol{\Theta}_{N-1}^{(n+1)} 
    + d_1 \boldsymbol{\Theta}_{N-1}^{(n-1)}, \label{def:Phi} 
\end{align}
where $\boldsymbol\Theta_{N-1}^{(n)}, \boldsymbol\Theta_{N-1}^{(n+1)}$, and $\boldsymbol\Theta_{N-1}^{(n-1)}$ are elementary Toeplitz matrices as defined in (\ref{def:Elementary_Toeplitz_matrix}). Then, the optimization problem (\ref{P2_2}) is equivalent to the following SDP:
\begin{subequations} \label{P_LMI_0}
\begin{align} 
    &\underset{\substack{t \in \mathbb{R}, \  {\bf x} \in {\mathbb C}^N \\{\bf h}^{(1)}, \ldots, {\bf h}^{(L)} \in {\mathbb C}^N ,\\{\bf Q}^{(1)}, \ldots, {\bf Q}^{(L)} \in {\mathbb H}_+^{N}, \\{\bf P}^{(1)}, \ldots, {\bf P}^{(L)} \in {\mathbb H}_+^{N-1}}}{\mathrm{minimize}} \hspace*{2mm} t \label{P_LMI_0_a}\\
    & \hspace*{-1mm} \mathrm{subject\ to} 
    \hspace*{2mm} t \delta_n = \mathrm{Tr}\left( \boldsymbol{\Theta}_{N}^{(n)}{\bf Q}^{(l)} \right) 
    + \mathrm{Tr}\Big({\bf \Phi}_{N-1}^{(n)} {\bf P}^{(l)} \Big), \nonumber \\ 
    & \hspace*{39mm} \forall n \in \mathbb{Z}_N,\forall l \in \{1,...,L\} \label{P_LMI_0_b} \\
    &\hspace*{15mm} 
    \renewcommand{\arraystretch}{0.8}
    \left[\begin{array}{cc}
        \mathbf{Q}^{(l)} & \mathbf{h}^{(l)} \\
        \mathbf{h}^{(l)H} & 1
    \end{array}\right] \succeq {\bf O}_{N+1}, \forall l \in \{1,...,L\} \label{P_LMI_0_c} \\ 
    &\hspace*{16mm} h_n^{(l)} = x_nx_{n-l}^*, \forall n \in \mathbb{Z}_N, \forall l \in \{1,...,L\} \label{P_LMI_0_d}\\
    &\hspace*{16mm} |x_n| = 1, \forall n \in \mathbb{Z}_N, \label{P_LMI_0_e}
\end{align}
\end{subequations}
where 
$\delta_n$ is defined in \eqref{def:delta},
$\boldsymbol\Theta_{N}^{(n)}$ is defined in \eqref{def:Elementary_Toeplitz_matrix},
and $h_n^{(l)}$ is the $n$th component of ${\bf h}^{(l)}$.
\begin{theorem} \label{thm:SIP_LMI}      
    The optimization problem \eqref{P2_2}, an SIP with infinite constraints, is equivalent to the optimization problem \eqref{P_LMI_0}, an SDP with a finite number of constraints.
\end{theorem}
\begin{proof}
    The proof of Theorem \ref{thm:SIP_LMI} is presented in Appendix \ref{appendix:proof_thm_SIP_LMI}.  
\end{proof}

By Theorem \ref{thm:SIP_LMI}, the infinite constraints in \eqref{P2_2_b} are equivalent to a finite number of LMIs in \eqref{P_LMI_0_b} and \eqref{P_LMI_0_c}.
In addition, \eqref{P_LMI_0_b} is an affine combination of positive semidefinite (PSD) matrices. 
Moreover, \eqref{P_LMI_0_c} can be expressed as \mbox{$\mathbf{Q}^{(l)} \succeq {\bf h}^{(l)}{\bf h}^{(l)H}$} by the Schur complement.

However, the optimization problem \eqref{P_LMI_0} is still non-convex due to the unimodular constraints in (\ref{P_LMI_0_e}).
To tackle the non-convex constraints, we intend to reformulate the optimization problem \eqref{P_LMI_0} into an SDP with a rank-one constraint. 
We introduce a new variable 
\begin{align} \label{def:X}
    {\bf X} = {\bf x}{\bf x}^H \in \mathbb{H}_+^N,
\end{align}
and notice that \eqref{def:X} is equivalent to ${\bf X}$ being a rank-one PSD matrix.
Then, from (\ref{P_LMI_0_d}), we can have
\begin{align} \label{bf_h}
    {\bf h}^{(l)} 
    = \mathrm{diag}\left( \boldsymbol{\Theta}_{N}^{(l)}{\bf X}
\right), 
\end{align}
where $\boldsymbol{\Theta}_{N}^{(l)}$ is defined in (\ref{def:Elementary_Toeplitz_matrix}).
Moreover, the unimodular constraints in (\ref{P_LMI_0_e}) can be expressed as:
\begin{align} \label{Tr(EX)}
    |x_n|^2
    &= (({\bf e}^{(n)})^T{\bf x})(({\bf e}^{(n)})^T{\bf x})^H
    = ({\bf e}^{(n)})^T{\bf x}{\bf x}^H{\bf e}^{(n)} \nonumber \\
    &= \mathrm{Tr}\left({\bf e}^{(n)}({\bf e}^{(n)})^T{\bf X}\right) 
    = \mathrm{Tr}({\bf E}^{(n)}{\bf X})
    = 1,
\end{align} 
where ${\bf e}^{(n)}$ is the $n$-th $N$-dimensional standard unit vector defined as having $1$ at position $n$ and $0$ elsewhere,
and \mbox{${\bf E}^{(n)} = {\bf e}^{(n)}({\bf e}^{(n)})^T$}.
Substituting \eqref{P_LMI_0_d} with \eqref{bf_h}, \eqref{P_LMI_0_e} with \eqref{Tr(EX)}, and including the rank-one constraint, the equivalent problem of the optimization problem (\ref{P_LMI_0}) is given by the following.
\begin{subequations} \label{P_LMI_1}
    \begin{align} 
        &\underset{\substack{t \in \mathbb{R}, \ {\bf X} \in {\mathbb H}_+^N \\{\bf h}^{(1)}, \ldots, {\bf h}^{(L)} \in {\mathbb C}^N ,\\{\bf Q}^{(1)}, \ldots, {\bf Q}^{(L)} \in {\mathbb H}_+^{N}, \\{\bf P}^{(1)}, \ldots, {\bf P}^{(L)} \in {\mathbb H}_+^{N-1}}}{\mathrm{minimize}} \hspace*{2mm} t\\
        & \hspace*{-3mm} \mathrm{subject\ to} 
        \hspace*{2mm} t \delta_n = \mathrm{Tr}\left( \boldsymbol{\Theta}_{N}^{(n)}{\bf Q}^{(l)} \right) 
        + \mathrm{Tr}\Big({\bf \Phi}_{N-1}^{(n)} {\bf P}^{(l)} \Big), \nonumber \\ 
        & \hspace*{40mm} \forall n \in \mathbb{Z}_N, \forall l \in \{1,...,L\} \\
        &\hspace*{14mm} 
        \renewcommand{\arraystretch}{0.8}
        \left[\begin{array}{cc}
            \mathbf{Q}^{(l)} & \mathbf{h}^{(l)}\\
            \mathbf{h}^{(l)H} & 1
        \end{array}\right] \succeq {\bf O}_{N+1}, \forall l \in \{1,...,L\} \\ 
        &\hspace*{15mm} {\bf h}^{(l)} =\mathrm{diag}\left( \boldsymbol{\Theta}_{N}^{(l)}{\bf X}
\right) , \forall l \in \{1,...,L\} \\ 
        &\hspace*{15mm} \mathrm{Tr}({\bf E}^{(n)}{\bf X}) = 1, \forall n \in \mathbb{Z}_N \\
        &\hspace*{15mm} \mathrm{rank}({\bf X}) = 1. \label{P_LMI_1_f}
    \end{align}
\end{subequations}

The optimization problem \eqref{P_LMI_1} is an SDP with the rank-one constraint \eqref{P_LMI_1_f}, which is still a non-convex optimization problem.
Algorithms that were proposed to deal with SDPs with rank-one constraints can be roughly categorized into three approaches: 1) semidefinite relaxation (SDR) \cite{Xu2014,Huang2010,Zhang2000,Luo2010,Wang2013,Wiesel2005,Ma2002,Lu2002}; 2) rank minimization \cite{Chitre20,Huang23,Fazel_reweighted_minimization,Dattorro19,Lin2024,Gu_2014_CVPR}; 3) sequential rank-one constraint relaxation (SROCR) \cite{Cao17}.
When applying SDR, the rank-one constraint in the original problem is relaxed.
For a special subclass of problems, the SDR yields an optimal solution that has a rank-one solution \cite{Xu2014,Huang2010,Zhang2000}.
However, in many other cases \cite{Luo2010,Wang2013,Wiesel2005,Ma2002,Lu2002}, SDR does not guarantee a rank-one solution, and the obtained solution is usually suboptimal.
In the second category, a penalty term for rank minimization is added to the objective function, such as using the trace heuristic \cite{Chitre20,Huang23,Fazel_reweighted_minimization,Dattorro19,Lin2024}, the nuclear norm heuristic \cite{Fazel_reweighted_minimization,Gu_2014_CVPR}, and the log-det heuristic \cite{Fazel_reweighted_minimization}.
The challenge of this approach is that its performance depends on the quality of the chosen rank-function approximation and the selection of penalty parameters, which is often intractable \cite{Cao17}.
In the third category, the SROCR algorithm proposed in \cite{Cao17} gradually strengthens the rank-one constraint over sequential iterations.
This approach not only facilitates finding a feasible solution, but also tends to achieve a smaller objective value in minimization problems \cite{Cao17}, as it avoids handling a penalty term for rank minimization in the objective function.
In light of the above advantages, we choose to apply the SROCR algorithm \cite{Cao17} to deal with the optimization problem \eqref{P_LMI_1}.

To handle the rank-one constraint in \eqref{P_LMI_1_f}, we present its reformulation as follows \cite{Cao17}.
Given \mbox{${\bf X}\in{\mathbb{H}}_+^N$},
we denote \mbox{$\lambda_{\rm max}({\bf X})$} as the largest eigenvalue of \mbox{${\bf X}$}.
Since \mbox{${\rm Tr}({\bf X})=\sum_{i=0}^{N-1}\lambda_i$} \cite{Horn1985},
the constraint \mbox{${\rm rank}({\bf X})=1$} is equivalent to \cite{Cao17}
\begin{equation}  \label{eq:rank1_equivalent_1} 
    \lambda_{\rm max}({\bf X})={\rm Tr}({\bf X}).
\end{equation}
Moreover, according to Rayleigh-Ritz theorem \cite{Horn1985}, we have 
\begin{equation} \label{eq:rank1_equivalent_2_1} 
\begin{aligned} 
    \lambda_{\rm max}({\bf X})
    &= \mathop{\rm{max}}\limits_{{\bf v}\in\mathbb{C}^{N}} \frac{{\bf v}^H {\bf X} {\bf v}}{{\bf v}^H{\bf v}}. 
\end{aligned}    
\end{equation}
In addition, denote ${\bf u}_{\rm max}({\bf X})$ as the principal eigenvector corresponding to $\lambda_{\rm max}({\bf X})$. 
Without loss of generality, we assume $\|{\bf u}_{\rm max}({\bf X})\|_2=1$. 
Then, from \eqref{eq:rank1_equivalent_2_1}, we obtain
\begin{equation} \label{eq:rank1_equivalent_2_2}
\begin{aligned}
    \lambda_{\rm max}({\bf X})
    &= {\bf u}_{\rm max}\left( {\bf X} \right)^H {\bf X} {\bf u}_{\rm max}\left( {\bf X} \right).
\end{aligned}    
\end{equation}
With \eqref{eq:rank1_equivalent_1} and \eqref{eq:rank1_equivalent_2_2}, the rank-one constraint in (\ref{P_LMI_1_f}) is equivalent to the constraint (\ref{opt: rank-1 constraint of problem during SROCR}) in the following problem.
\begin{subequations} \label{P_SROCR}
\begin{align} 
    &\underset{\substack{t \in \mathbb{R}, \ {\bf X} \in {\mathbb H}_+^N \\{\bf h}^{(1)}, \ldots, {\bf h}^{(L)} \in {\mathbb C}^N ,\\{\bf Q}^{(1)}, \ldots, {\bf Q}^{(L)} \in {\mathbb H}_+^{N}, \\{\bf P}^{(1)}, \ldots, {\bf P}^{(L)} \in {\mathbb H}_+^{N-1}}}{\mathrm{minimize}} \hspace*{2mm} t\\
    & \hspace*{-3mm} \mathrm{subject\ to} 
    \hspace*{2mm} t \delta_n = \mathrm{Tr}\left( \boldsymbol{\Theta}_{N}^{(n)}{\bf Q}^{(l)} \right) 
    + \mathrm{Tr}\Big({\bf \Phi}_{N-1}^{(n)} {\bf P}^{(l)} \Big), \nonumber \\ 
    & \hspace*{40mm} \forall n \in \mathbb{Z}_N, \forall l \in \{1,...,L\}\\
    &\hspace*{13mm} 
    \renewcommand{\arraystretch}{0.8}
    \left[\begin{array}{cc}
        \mathbf{Q}^{(l)} & \mathbf{h}^{(l)}\\
        \mathbf{h}^{(l)H} & 1
    \end{array}\right] \succeq {\bf O}_{N+1}, \forall l \in \{1,...,L\} \\ \label{opt: delay-multiply constraint of problem during SROCR}
    &\hspace*{14mm} {\bf h}^{(l)} =\mathrm{diag}\left( \boldsymbol{\Theta}_{N}^{(l)}{\bf X}\right), \forall l \in \{1,...,L\}\\ \label{P_SROCR_TrX}
    &\hspace*{14mm} \mathrm{Tr}({\bf E}^{(n)}{\bf X}) = 1, \forall n \in \mathbb{Z}_N \\ 
    &\hspace*{14mm} {\bf u}_{\max}\left({\bf X}\right)^H{\bf X}{\bf u}_{\max}\left({\bf X}\right) = \mathrm{Tr}\left({\bf X}\right). \label{opt: rank-1 constraint of problem during SROCR} 
\end{align}
\end{subequations}
However, the problem is still non-convex since (\ref{opt: rank-1 constraint of problem during SROCR}) implies the rank-one constraint. 
The idea of the SROCR algorithm \cite{Cao17} is to sequentially solve the problem with the partially relaxed constraint in \eqref{ineq:SROCR_constraint} instead of using \eqref{opt: rank-1 constraint of problem during SROCR}:
\begin{equation} \label{ineq:SROCR_constraint}
    {\bf u}_{\max}({\bf X}^{\langle i \rangle})^H{\bf X}{\bf u}_{\max}({\bf X}^{\langle i \rangle}) \geq w^{\langle i \rangle}\mathrm{Tr}({\bf X}),
\end{equation}
where \mbox{${\bf X}^{\langle i \rangle}$} is the feasible solution obtained in the $i$-th iteration, \mbox{$i=0,1,2,...$}, and \mbox{$w^{\langle i \rangle} \in (0,1]$} is the parameter for rank-one constraint approximation. 
Moreover, since the constraint \eqref{P_SROCR_TrX} implies \mbox{$\mathrm{Tr}({\bf X})=N$}, the inequality \eqref{ineq:SROCR_constraint} can be simplified to
\begin{equation} \label{ineq:SROCR_constraint_N}
    {\bf u}_{\max}({\bf X}^{\langle i \rangle})^H{\bf X}{\bf u}_{\max}({\bf X}^{\langle i \rangle}) \geq w^{\langle i \rangle}N.
\end{equation}
We increase \mbox{$w^{\langle i \rangle}$} in each iteration until it approaches one, thereby satisfying \eqref{opt: rank-1 constraint of problem during SROCR}. 
Moreover, it is worth noting that \mbox{$w^{\langle i \rangle}$} satisfies: \cite{Cao17} 
\begin{equation}
    w^{\langle i \rangle} \leq \dfrac{{\bf u}_{\max}({\bf X}^{\langle i \rangle})^H{\bf X}{\bf u}_{\max}({\bf X}^{\langle i \rangle})}{N} \leq \dfrac{\lambda_{\max}({\bf X})}{N},
\end{equation}
where the first $``\leq"$ is from \eqref{ineq:SROCR_constraint_N}, and the second $``\leq"$ is due to \eqref{eq:rank1_equivalent_2_2}. 
Replace (\ref{opt: rank-1 constraint of problem during SROCR}) with (\ref{ineq:SROCR_constraint_N}) and given $\{ {\bf X}^{\langle i \rangle}, w^{\langle i \rangle} \}$, the optimization problem to be solved in the $i+1$-th iteration becomes
\begin{subequations} \label{P:SROCR_iter}
    \begin{align} 
        &\underset{\substack{t \in \mathbb{R}, \ {\bf X} \in {\mathbb H}_+^N\\{\bf h}^{(1)}, \ldots, {\bf h}^{(L)} \in {\mathbb C}^N ,\\{\bf Q}^{(1)}, \ldots, {\bf Q}^{(L)} \in {\mathbb H}_+^{N}, \\{\bf P}^{(1)}, \ldots, {\bf P}^{(L)} \in {\mathbb H}_+^{N-1}}}{\mathrm{minimize}} \hspace*{2mm} t \label{P:SROCR_iter_a}\\
        & \hspace*{-3mm} \mathrm{subject\ to} 
        \hspace*{2mm} t \delta_n = \mathrm{Tr}\left( \boldsymbol{\Theta}_{N}^{(n)}{\bf Q}^{(l)} \right) 
        + \mathrm{Tr}\Big({\bf \Phi}_{N-1}^{(n)} {\bf P}^{(l)} \Big), \nonumber \\ 
        &\hspace*{40mm} \forall n \in \mathbb{Z}_N, \forall l \in \{1,...,L\} \label{P:SROCR_iter_b}\\
        &\hspace*{13mm} 
        \renewcommand{\arraystretch}{0.8}
        \left[\begin{array}{cc}
            \mathbf{Q}^{(l)} & \mathbf{h}^{(l)}\\
            \mathbf{h}^{(l)H} & 1
        \end{array}\right] \succeq {\bf O}_{N+1}, \forall l \in \{1,...,L\} \label{P:SROCR_iter_c} \\ 
        &\hspace*{14mm} {\bf h}^{(l)} =\mathrm{diag}\left( \boldsymbol{\Theta}_{N}^{(l)}{\bf X}
\right), \forall l \in \{1,...,L\} \label{P:SROCR_iter_d} \\ 
        &\hspace*{14mm} \mathrm{Tr}({\bf E}^{(n)}{\bf X}) = 1, \forall n \in \mathbb{Z}_N \label{P:SROCR_iter_e} \\ 
        &\hspace*{14mm} {\bf u}_{\max}\left({\bf X}^{\langle i \rangle}\right)^H{\bf X}{\bf u}_{\max}\left({\bf X}^{\langle i \rangle}\right) \geq w^{\langle i \rangle}N. \label{P:SROCR_iter_f}
    \end{align}
\end{subequations}

The optimization problem (\ref{P:SROCR_iter}) is a convex optimization problem that can be solved through the $\rm{CVX}$ toolbox \cite{CVX}, a package for addressing convex optimization problems. 
The proposed method is summarized in Algorithm \ref{Alg:SROCR-TPSL}.
Firstly, the initial point \mbox{${\bf X}^{\langle 0 \rangle}$} is obtained by solving the optimization problem (\ref{P:SROCR_iter}) without considering \eqref{P:SROCR_iter_f}.
Moreover, we set 
\begin{align} \label{w_0}
    {w}^{\langle 0 \rangle} = \left(1 - \tfrac{ \lambda_{\max}({\bf X}^{\langle 0 \rangle})}{N} \right)/\zeta,
\end{align}
where $\zeta$ is a preset parameter for determining the increment of $w^{\langle i \rangle}$.
If a feasible solution ${\bf X}^*$ can be obtained for the optimization problem (\ref{P:SROCR_iter}) in the \mbox{${i+1}$}-th iteration, we then set \mbox{${\bf X}^{\langle i+1 \rangle} = {\bf X}^*$} and let \mbox{$w^{\langle i+1 \rangle}$} updated as 
\begin{align} \label{w_i}
    w^{\langle i+1 \rangle} = \tfrac{\lambda_{\max}\left({\bf X}^{\langle i+1 \rangle}\right)}{N} + \delta^{\langle i+1 \rangle},
\end{align}
where
\begin{align} \label{delta_i}
    \delta^{\langle i+1 \rangle} 
    = \left(1 - \tfrac{ \lambda_{\max}({\bf X}^{\langle i+1 \rangle})}{N} \right)/\zeta.
\end{align}
If the increment \mbox{$\delta^{\langle i+1 \rangle}$} is too large and causes the optimization problem (\ref{P:SROCR_iter}) to become infeasible in the \mbox{${i+1}$}-th iteration, \mbox{$\delta^{\langle i \rangle}$} will be halved and we retry to solve the optimization problem (\ref{P:SROCR_iter}) again. 
The iteration will stop when \mbox{$w^{\langle i \rangle}$} exceeds $\kappa$, a value smaller than and close to 1, and the objective values expressed in decibels (${\rm dB}$) differ by less than $\epsilon$ between consecutive iterations. 
Lastly, perform the eigendecomposition on \mbox{${\bf X}^{\langle i+1 \rangle}$}, extract its largest eigenvalue and principle eigenvector, and obtain the optimal sequence \mbox{${\bf x}_{\rm opt}$}.

\begin{algorithm}[ht!]
\begin{small}
\caption{Proposed Method for DTAF Optimization Problem with PSL Minimization under Unimodular Constraints}
\label{Alg:SROCR-TPSL}
\textbf{Input: } $N$, $L$, $f_R$, $\zeta$, $\kappa$, $\epsilon$\\ 
\textbf{Output: } ${\bf x}_{\rm opt}$ \\
\begin{algorithmic}[1]
    \vspace{-2mm}
    \State Solve the optimization problem (\ref{P:SROCR_iter}) without considering \eqref{P:SROCR_iter_f} and obtain the initial point \mbox{${\bf X}^{\langle 0 \rangle}$}. 
    \State Set \mbox{$w^{\langle 0 \rangle}$} according to \eqref{w_0}.
    \State \mbox{$i \gets 0$}.
    \Repeat
    \State Given \mbox{$\{ {\bf X}^{\langle i \rangle}, w^{\langle i \rangle} \}$}, solve the optimization problem (\ref{P:SROCR_iter}).
    \If {the optimization problem (\ref{P:SROCR_iter}) is feasible}
    \State Obtain the feasible solution \mbox{${\bf X}^{\langle i+1 \rangle} \gets {\bf X}^*$}.
    \State Obtain the objective value \mbox{$t^{\langle i+1 \rangle} \gets t^*$}.
    \State Set \mbox{$\delta^{\langle i+1 \rangle}$ according to \eqref{delta_i}}.
    \Else
    \State \mbox{${\bf X}^{\langle i+1 \rangle} \gets {\bf X}^{\langle i \rangle}$}.
    \State \mbox{$\delta^{\langle i+1 \rangle} \gets \delta^{\langle i \rangle}/2$}.
    \EndIf
    \State \mbox{$w^{\langle i+1 \rangle} \gets \tfrac{\lambda_{\max}\left({\bf X}^{\langle i+1 \rangle}\right)}{N} + \delta^{\langle i+1 \rangle}$}.
    \State \mbox{$i \gets i+1$}.
    \Until 
    \mbox{$w^{\langle i-1 \rangle} \geq \kappa$} and 
    \mbox{$\left|10\log_{10}\left( \tfrac{t^{\langle i \rangle}}{t^{\langle i-1 \rangle}} \right) \right| \leq \epsilon$}.
    \State Perform eigendecomposition on \mbox{${\bf X}^{{\langle i \rangle}}$} and extract its principal eigenvector, 
    ${\bf u}_{\max}({\bf X}^{{\langle i \rangle}})$,
    along with the largest eigenvalue, 
    ${\lambda}_{\max}({\bf X}^{{\langle i \rangle}})$. 
    \State Obtain \mbox{${\bf x}_{\rm opt} = \sqrt{{\lambda}_{\max}({\bf X}^{{\langle i \rangle}})}{\bf u}_{\max}({\bf X}^{{\langle i \rangle}})$}.  
\end{algorithmic}
\end{small}
\end{algorithm}

\section{Simulation results} \label{sec:numerical_results}
\subsection{Definition of Evaluation Metrics}

The normalized true peak sidelobe level (NTPSL), the normalized grid peak sidelobe level (NGPSL), and the normalized weighted integrated sidelobe level (NWISL) are defined to evaluate the PSL and WISL suppression performance of the waveform. 
Firstly, the NTPSL, defined in (\ref{def:NTPSL}) and expressed in ${\rm dB}$, aims to evaluate the PSL suppression performance of the waveform in a continuous Doppler frequency shift ROI.
\begin{equation} \label{def:NTPSL} 
\begin{aligned}
    \mathrm{NTPSL} = 20\log_{10}\left( \underset{(l,f_D) \in \Gamma_s}{\max}  \dfrac{|A(l,f_D)|}{N} \right),
\end{aligned}
\end{equation}
where $A(l,f_D)$ is defined in (\ref{def:DTAF}), and $\Gamma_s$ is given by (\ref{eq:Gamma_s}): $\Gamma_s = \{\pm 1,\pm 2,..., \pm L\} \times [-f_R,f_R]$.
Given a sequence ${\bf x}$, the NTPSL in (\ref{def:NTPSL}) can be obtained through the bisection method \cite{Boyd_Vandenberghe_2004}
\footnote{
    Firstly, the maximum sidelobe level for each time delay bin \mbox{$l \in \{ \pm 1, \pm 2,...,\pm L \}$} is determined over the continuous Doppler frequency shift ROI \mbox{$[f_R,f_R]$}, using the bisection method \cite{Boyd_Vandenberghe_2004} [Sec 4.2.5]. 
    The NTPSL is then obtained by selecting the largest value among these $2L$ sidelobe levels. 
}.
Secondly, the NGPSL, defined in (\ref{def:NGPSL}) and expressed in ${\rm dB}$, evaluates the normalized PSL in the sidelobe region $\Psi_s$ defined in (\ref{def:Psi_s}) that considers grid points in the delay-Doppler plane  \cite{Jing19} 
\begin{equation} \label{def:NGPSL} 
\begin{aligned}
    \mathrm{NGPSL} =20\log_{10}\left( \underset{(l,f_D) \in {\Psi}_s}{\max}  \dfrac{|A_d(l ,k; M)|}{N} \right),
\end{aligned}
\end{equation}
where $A_d(l ,k; M)$ is defined in \eqref{def:DAF},
\begin{equation} \label{def:Psi_s} 
\begin{aligned}
    {\Psi}_s &= \{\pm 1,\pm 2,...,\pm L\} \times \{ 0, \tfrac{\pm 1}{M}, \tfrac{\pm 2}{M}, \ldots, \tfrac{\pm K}{M}\}, 
\end{aligned}
\end{equation}
$K$ is an integer, and \mbox{$\tfrac{K}{M} = f_R$}.
Thirdly, the normalized weighted integrated sidelobe level (NWISL) expressed in ${\rm dB}$ is defined as
\begin{equation} \label{eq:NWISL}
    \mathrm{NWISL} = 20\log_{10} \left(\frac{1}{L|\Omega_D|} \sum_{l = 1}^{L}\sum_{f_D\in\Omega_D} w_l \left(\frac{|A(l ,f_D)|}{N}\right)\right),
\end{equation}
where $w_l$ is the weight of time delay bins, {$\Omega_D=\{ 0,\tfrac{\pm 1}{M},\tfrac{\pm 2}{M},....,\tfrac{\pm K}{M}\}$}, and \mbox{$|\Omega_D|=2K+1$}.

\subsection{Numerical Results} \label{sec:Numerical_Ex}
Numerical results\footnote{All simulations are conducted on a personal computer with 4.70 GHz AMD Ryzen 9 7900X, 128 GB.} are carried out to demonstrate the advantages of the proposed algorithm \ref{Alg:SROCR-TPSL} (labeled as ``Proposed Algorithm") in achieving unimodular sequences that exhibit low PSL across a continuous Doppler frequency shift ROI.
Two algorithms are compared with the proposed algorithm, one is the generalized maximum block improvement (GMBI)-$P^{ac}$ algorithm proposed in \cite{Chen22} (labeled as ``GMBI-$P^{ac}$") and another is the iterative sequential quartic optimization (ISQO) algorithm proposed in \cite{Yang18} (labeled as ``ISQO").
Applying the GMBI-$P^{ac}$ algorithm \cite{Chen22}, 
a set of $J$ unimodular sequences can be designed considering the PSL suppression for auto-AF and cross-AF in the problem $P^{ac}$ of \cite{Chen22}.
Since the scope of this work does not address sequence set design, we set \mbox{$J=1$} in the GMBI-$P^{ac}$ algorithm. 
With this setting, the GMBI-$P^{ac}$ algorithm reduces to solving a DAF optimization problem that minimizes PSL under unimodular constraints, which we compare with our proposed method.
Using the ISQO algorithm \cite{Yang18}, the DAF optimization problem with WISL minimization considering spectral nulling and peak-to-average-ratio (PAR) constraints was tackled. 
When the parameters are set to \mbox{$\beta=1$} and \mbox{$\gamma=1$} in the ISQO algorithm described in \cite{Yang18}, the algorithm turns to solve a DAF optimization problem that minimizes WISL under unimodular constraints which are compared with our proposed method.

\subsubsection{Case 1} \label{sim:case_1}
Unimodular sequences with \mbox{$N = 32$} are designed.
Algorithm \ref{Alg:SROCR-TPSL} is applied to consider PSL suppression over the sidelobe region \mbox{$\{ \pm 1,\pm 2,\pm 3 \}\times \left[ -\tfrac{3}{32}, \tfrac{3}{32} \right]$},
where the continuous Doppler frequency shift ROI is taken into account.
Additionally, we set 
\mbox{$\zeta = 10$},
\mbox{$\kappa=0.99$}, and \mbox{$\epsilon=0.001$}.
The selection of the $\zeta$ settings is discussed in Section \ref{subsubsec:zeta_settings}.
On the other hand, both the GMBI-$P^{ac}$ algorithm and the ISQO algorithm consider the sidelobe region over
\mbox{$\{ \pm 1,\pm 2,\pm 3 \} \times \{ 0, \tfrac{\pm 1}{32}, \tfrac{\pm 2}{32}, \tfrac{\pm 3}{32}\}$},
where the discrete delay-Doppler bins are considered.
Parameters in the GMBI-$P^{ac}$ algorithm are set as \mbox{$J=1$}, \mbox{$\rho=0.015$}, \mbox{$\mu=12$}, \mbox{$T_{\rm mbi}=100$}, and \mbox{$T_0=10000$}.
In addition, parameters in the ISQO algorithm are set as \mbox{$\beta=1$}, \mbox{$\gamma=1$}, 
\mbox{$\varepsilon=10^{-6}$}, and \mbox{$w_l=1,\forall l \in \{\pm 1,\pm 2,\pm 3\}$}. 
In each iteration of Algorithm \ref{Alg:SROCR-TPSL}, 
the values of $\delta^{\langle i \rangle}$, 
$w^{\langle i \rangle}$, and the normalized objective function values of the problem (\ref{P:SROCR_iter}), \mbox{$20\log_{10}(\tfrac{\sqrt{t^{\langle i \rangle}}}{N})$}, are shown in Fig.\ref{Fig:N_32_SROCR_parmameters}.
Moreover, the normalized DTAF results, \mbox{$20\log_{10}(\tfrac{\sqrt{A(l,f_D))}}{N})$}, obtained by Algorithm \ref{Alg:SROCR-TPSL}, the GMBI-$P^{ac}$ algorithm and the ISQO algorithm are demonstrated in Figs. \ref{Fig:N_32_SROCR_a}, \ref{Fig:N_32_GMBI_a}, and \ref{Fig:N_32_ISQO_a}.
In Figs. \ref{Fig:N_32_SROCR_b}, \ref{Fig:N_32_GMBI_b}, and \ref{Fig:N_32_ISQO_b}, the DTAF of each lag is shown.
More specifically, we plot the normalized \mbox{$A(1,f_D)$}, \mbox{$A(2,f_D)$}, and \mbox{$A(3,f_D)$} in Fig. \ref{Fig:N_32_SROCR_b}, and likewise in Figs. \ref{Fig:N_32_GMBI_b} and \ref{Fig:N_32_ISQO_b}.
It is shown in Fig. \ref{Fig:N_32_SROCR_b} that the proposed algorithm can achieve PSL suppression in the continuous Doppler frequency shift ROI, whereas the GMBI-$P^{ac}$ algorithm and the ISQO algorithm fail to do so. 
Since the GMBI-$P^{ac}$ algorithm and the ISQO algorithm only consider PSL suppression on the discrete delay Doppler frequency shift bins \mbox{$\tfrac{k}{32}$}, \mbox{$k\in\{\pm 3,\pm 2,\pm 1\}$}, some peaks may appear beyond the NGPSL as demonstrated in Figs. \ref{Fig:N_32_GMBI_b} and \ref{Fig:N_32_ISQO_b}.
In Table \ref{Table:N_32}, the NTPSL, NGPSL and NWISL values achieved by these waveforms are listed.
We can observe that the proposed algorithm achieves the lowest NTPSL (\mbox{$-29.30 \ {\rm dB}$}), with a considerable margin over the GMBI-$P^{ac}$ algorithm (\mbox{$-21.62 \ {\rm dB}$}) and the ISQO algorithm (\mbox{$-17.44 \ {\rm dB}$}).
In addition, the GMBI-$P^{ac}$ algorithm achieves the lowest NGPSL, while the ISQO algorithm achieves the lowest NWISL.

\begin{figure}[ht!]
    \centering
    \subfloat[]{\includegraphics[width=1.2in]{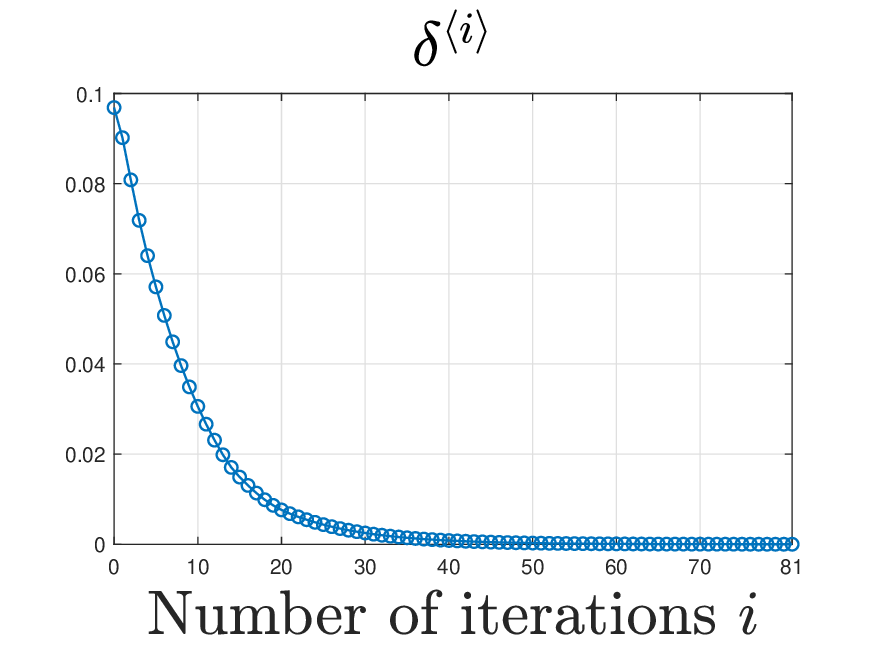} 
    \label{Fig:N_32_SROCR_delta}} 
    \hspace{-0.5cm}
    \subfloat[]{\includegraphics[width=1.2in]{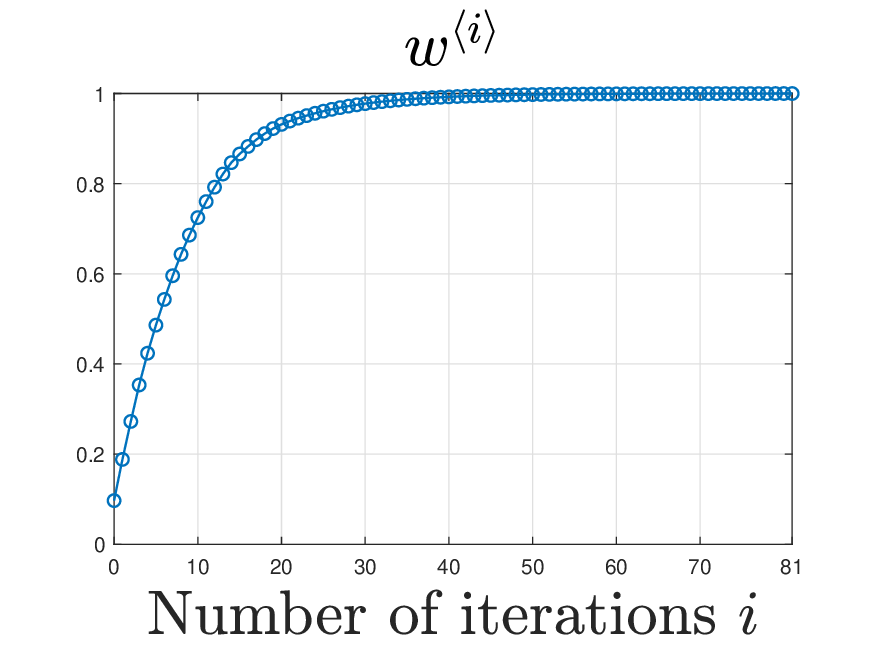}
    \label{Fig:N_32_SROCR_w}} 
    \hspace{-0.5cm}
    \subfloat[]{\includegraphics[width=1.2in]{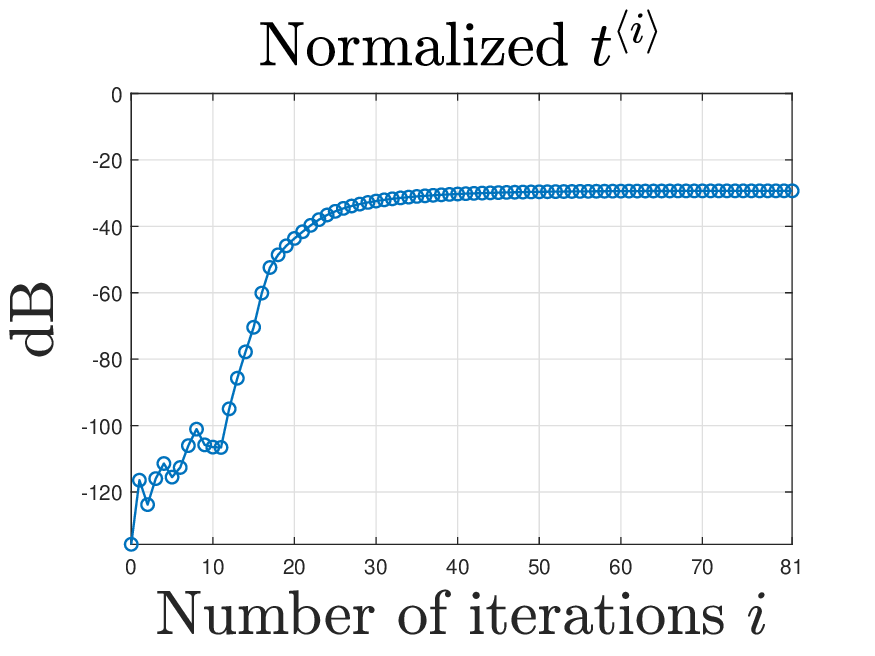} \label{Fig:N_32_SROCR_NTPSL}} 
    \caption{ The parameters of the proposed algorithm at each iteration in Case 1 (\mbox{$N=32$}): (a) $\delta^{\langle i \rangle}$, 
    (b) $w^{\langle i \rangle}$, and (c) normalized $t^{\langle i \rangle}$. }  
    \label{Fig:N_32_SROCR_parmameters}
\end{figure}

\begin{figure}[ht!]
    \centering
    \subfloat[]{\includegraphics[width=1.7in]{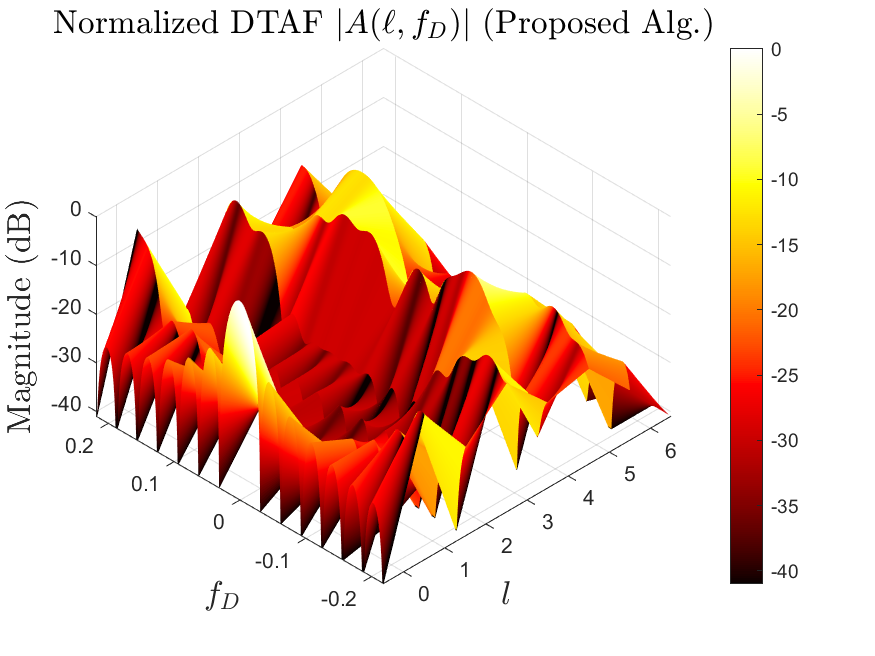} 
    \label{Fig:N_32_SROCR_a}} 
    \hspace{-0.5cm}  
    \subfloat[]{\includegraphics[width=1.8in]{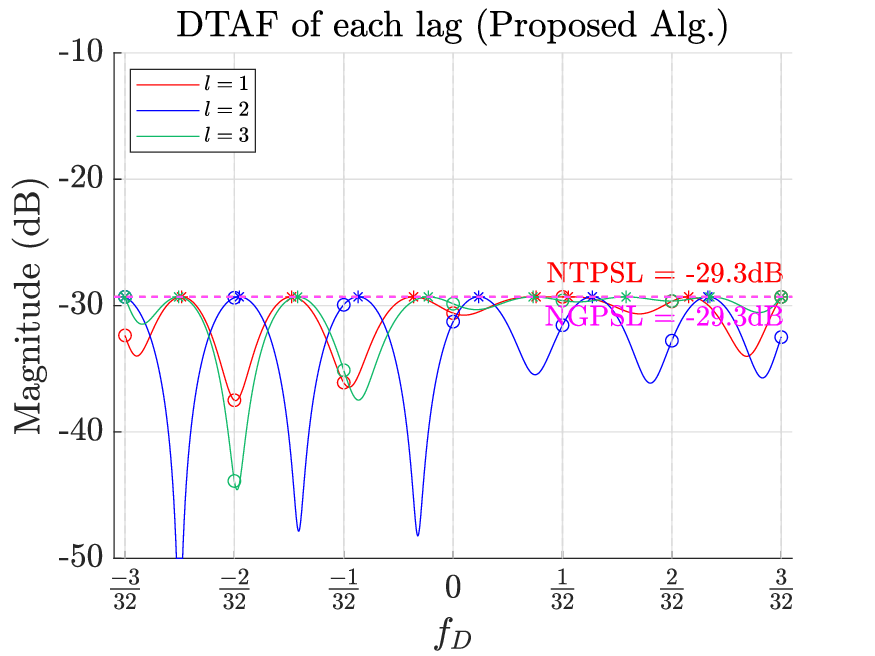}
    \label{Fig:N_32_SROCR_b}} 
    \par \vspace{-0.1mm}  
    \subfloat[]{\includegraphics[width=1.7in]{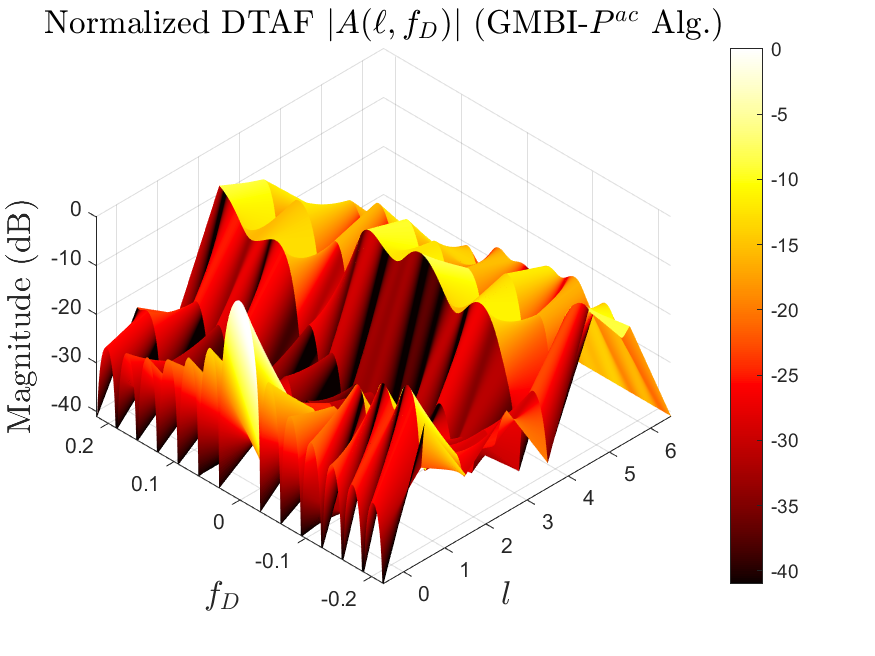}
    \label{Fig:N_32_GMBI_a}} 
    \hspace{-0.5cm} 
    \subfloat[]{\includegraphics[width=1.8in]{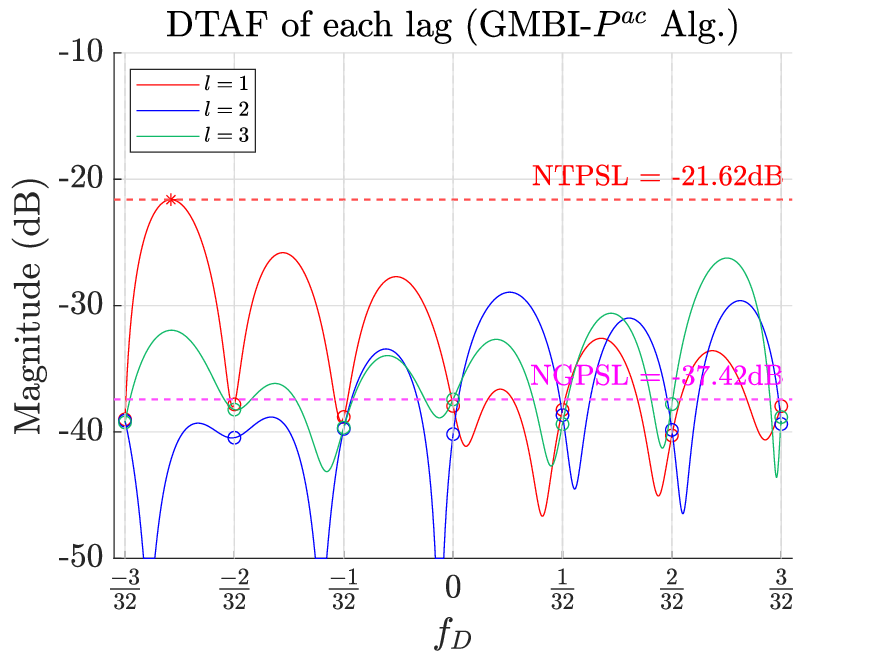}  
    \label{Fig:N_32_GMBI_b}} 
    \par \vspace{-0.1mm}  
    \subfloat[]{\includegraphics[width=1.7in]{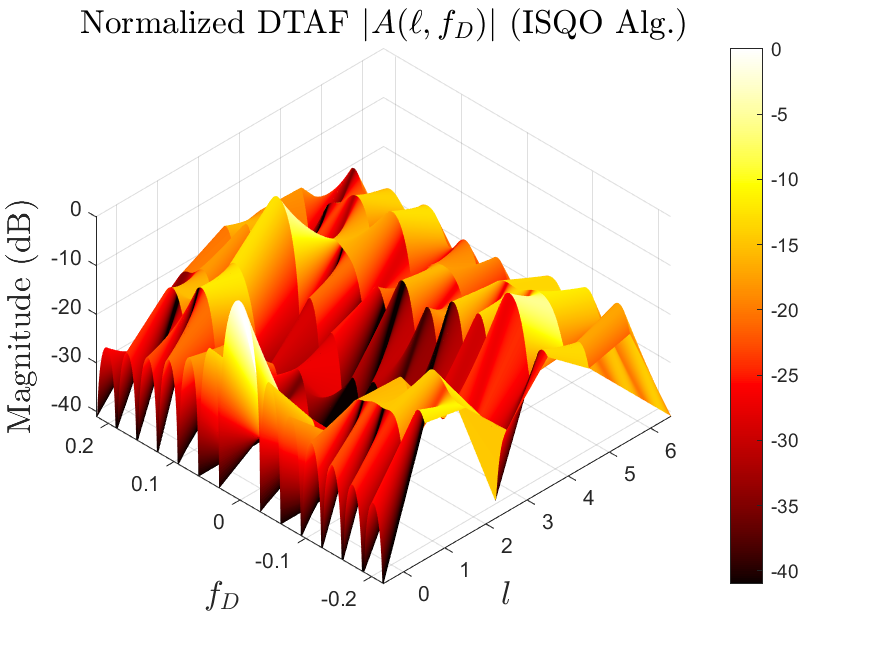}
    \label{Fig:N_32_ISQO_a}} 
    \hspace{-0.5cm} 
    \subfloat[]{\includegraphics[width=1.8in]{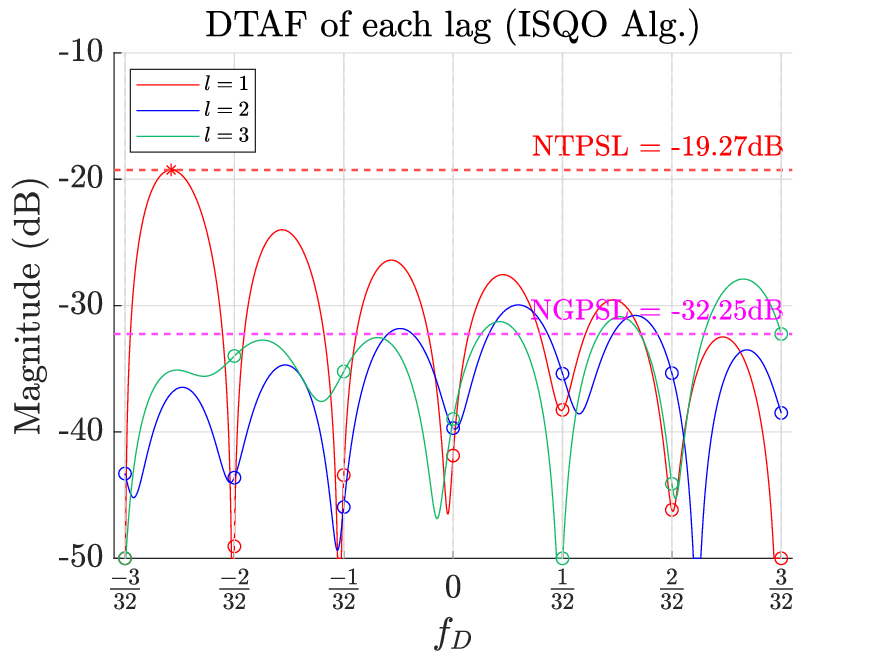}  
    \label{Fig:N_32_ISQO_b}} 
    \caption{ The DTAF and its magnitude of each lag obtained using different methods in Case 1 (\mbox{$N=32$}): (a) and (b) from the proposed algorithm, (c) and (d) from the GMBI-$P^{ac}$ algorithm \cite{Chen22}, (e) and (f) from the ISQO algorithm \cite{Yang18}. }  \label{Fig:N_32}  
\end{figure}

\begin{table}[ht!] 
    \centering
    \caption{NTPSL, NGPSL, and NWISL for different methods in Case 1 (\mbox{$N=32$}).}
    \label{Table:N_32}
    \begin{scriptsize}
    \begin{tabular}{  c  c  c  c  } 
        \hline
        \textbf{Method} & \textbf{NTPSL} (${\rm dB}$) & \textbf{NGPSL} (${\rm dB}$) & \textbf{NWISL} (${\rm dB}$) \\
        \hline
        Proposed Algorithm & $-29.30$ & $-29.30$ & $-31.25$ \\
        GMBI-$P^{ac}$ Algorithm \cite{Chen22} & $-21.62$ & $-37.42$ & $-38.91$ \\ 
        ISQO Algorithm \cite{Yang18} & $-19.27$ & $-32.25$ & $-40.47$ \\ 
        \hline
    \end{tabular}
    \end{scriptsize}
\end{table}

\subsubsection{Case 2}
Unimodular sequences with \mbox{$N = 64$} are designed. 
We apply Algorithm \ref{Alg:SROCR-TPSL} to consider the PSL suppression over the sidelobe region of
\mbox{$\{ \pm 1,\pm 2,..., \pm 6 \} \times \left[ -\tfrac{3}{64}, \tfrac{3}{64} \right]$},
and we set \mbox{$\zeta = 10$}, \mbox{$\kappa=0.99$}, and \mbox{$\epsilon=0.001$}.
While, both the GMBI-$P^{ac}$ algorithm and the ISQO algorithm consider the sidelobe region of
\mbox{$\{ \pm 1,\pm 2,...,\pm 6 \} \times \{ 0, \tfrac{\pm 1}{64}, \tfrac{\pm 2}{64}, \tfrac{\pm 3}{64}\}$}.
We have parameters in the GMBI-$P^{ac}$ algorithm set as \mbox{$J=1$}, \mbox{$\rho=0.01$}, \mbox{$\mu=100$}, \mbox{$T_{\rm mbi}=100$}, and \mbox{$T_0=10000$}.
Also, the ISQO algorithm is run with the parameters \mbox{$\beta=1$}, \mbox{$\gamma=1$}, 
\mbox{$\varepsilon=10^{-6}$}, and \mbox{$w_l=1,\forall l \in \{\pm 1,...,\pm 6\}$}.
In each iteration of Algorithm \ref{Alg:SROCR-TPSL}, 
the values of $\delta^{\langle i \rangle}$, 
$w^{\langle i \rangle}$, and the normalized objective function values of the optimization problem (\ref{P:SROCR_iter}) are shown in Fig.\ref{Fig:N_64_SROCR_parmameters}.
The normalized DTAF results of those sequences are shown in Figs. \ref{Fig:N_64_SROCR_a}, \ref{Fig:N_64_GMBI_a}, and \ref{Fig:N_64_ISQO_a}.
In Figs. \ref{Fig:N_64_SROCR_b}, \ref{Fig:N_64_GMBI_b}, and \ref{Fig:N_64_ISQO_b}, 
the DTAF of each lag is plotted.
It can be observed from Fig. \ref{Fig:N_32_SROCR_b} that the proposed algorithm can suppress PSL in the continuous Doppler frequency shift ROI, while the GMBI-$P^{ac}$ algorithm and the ISQO algorithm cannot. 
Additionally, the NTPSL, NGPSL, and NWISL values are summarized in Table \ref{Table:N_64}.
We can see that the proposed method achieves the lowest NTPSL compared to the GMBI-$P^{ac}$ algorithm and the ISQO algorithm.
Meanwhile, the GMBI-$P^{ac}$ algorithm attains the lowest NGPSL, while the ISQO algorithm achieves the lowest NWISL. 

\begin{figure}[ht!]
    \centering
    \subfloat[]{\includegraphics[width=1.2in]{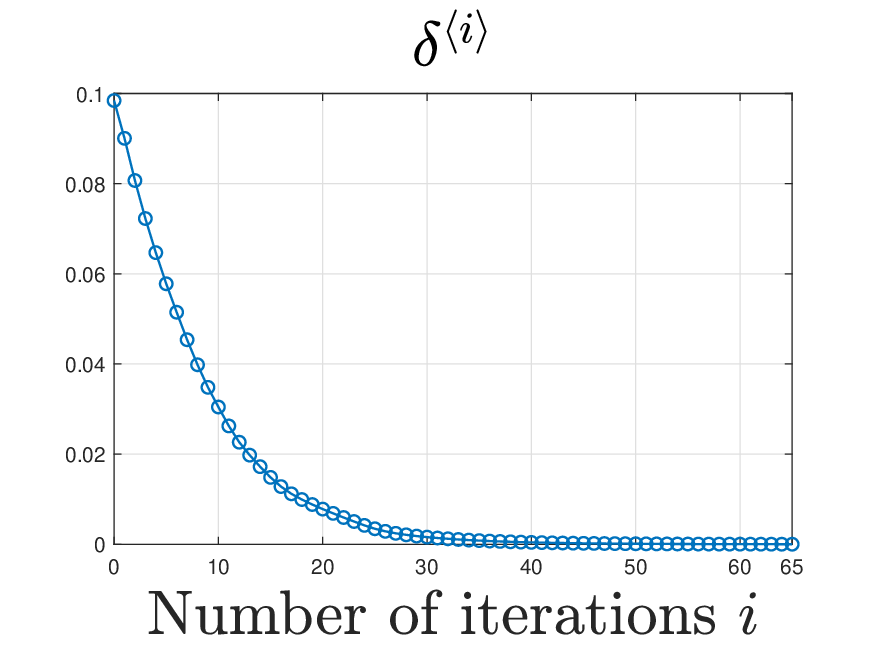} 
    \label{Fig:N_64_SROCR_delta}} 
    \hspace{-0.5cm}
    \subfloat[]{\includegraphics[width=1.2in]{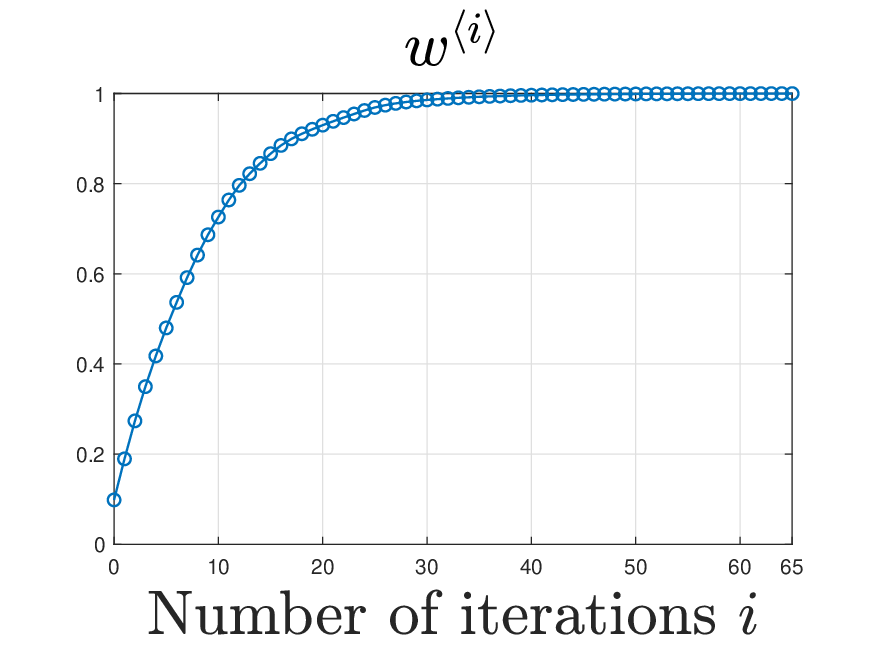}
    \label{Fig:N_64_SROCR_w}} 
    \hspace{-0.5cm}
    \subfloat[]{\includegraphics[width=1.2in]{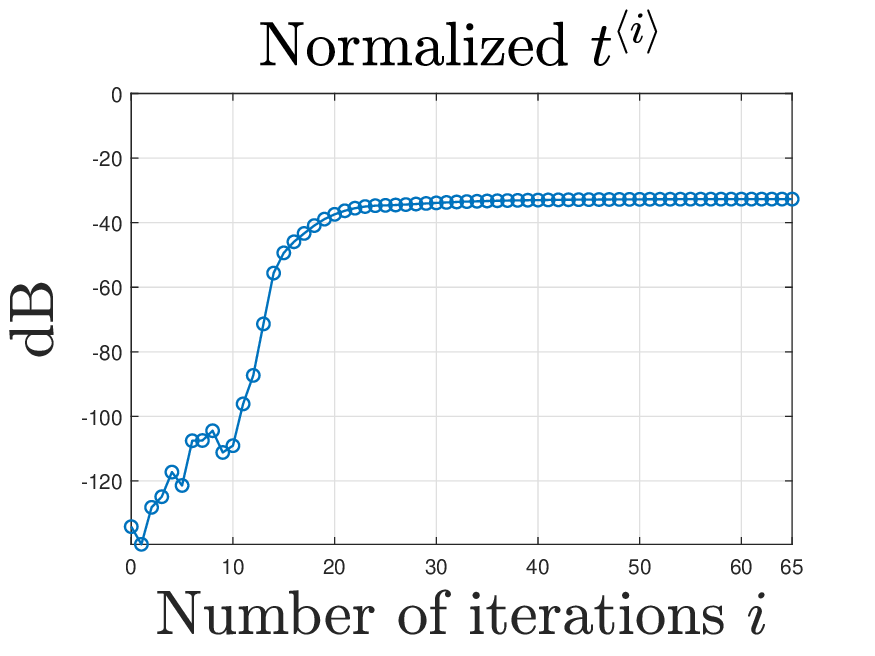} \label{Fig:N_64_SROCR_NTPSL}} 
    \caption{ The parameters of the proposed algorithm at each iteration in Case 2 (\mbox{$N=64$}): (a) $\delta^{\langle i \rangle}$, 
    (b) $w^{\langle i \rangle}$, and (c) normalized $t^{\langle i \rangle}$. }  
    \label{Fig:N_64_SROCR_parmameters}
\end{figure}

\begin{figure}[ht!]
    \centering
    \subfloat[]{\includegraphics[width=1.7in]{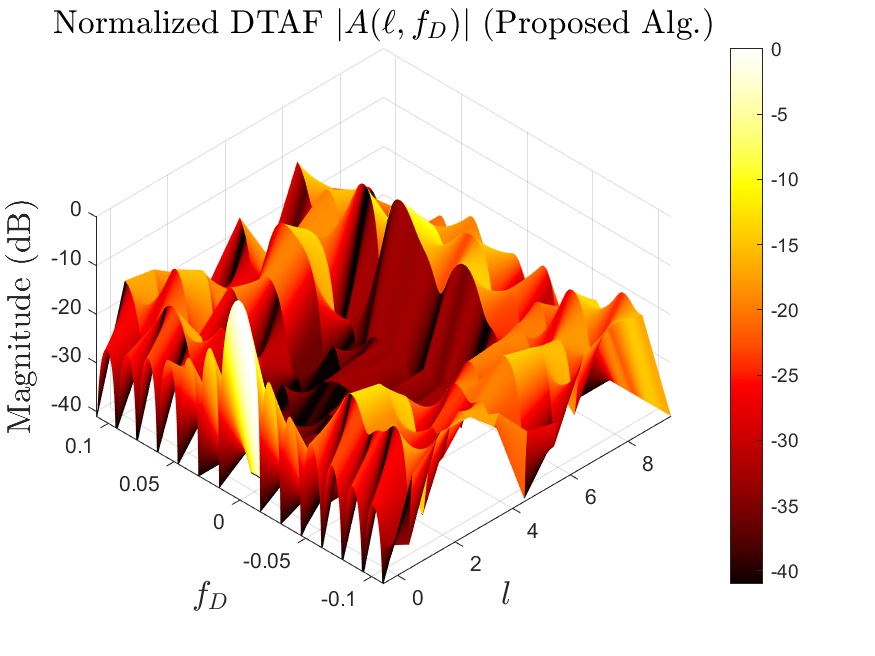} 
    \label{Fig:N_64_SROCR_a}} 
    \hspace{-0.5cm}  
    \subfloat[]{\includegraphics[width=1.8in]{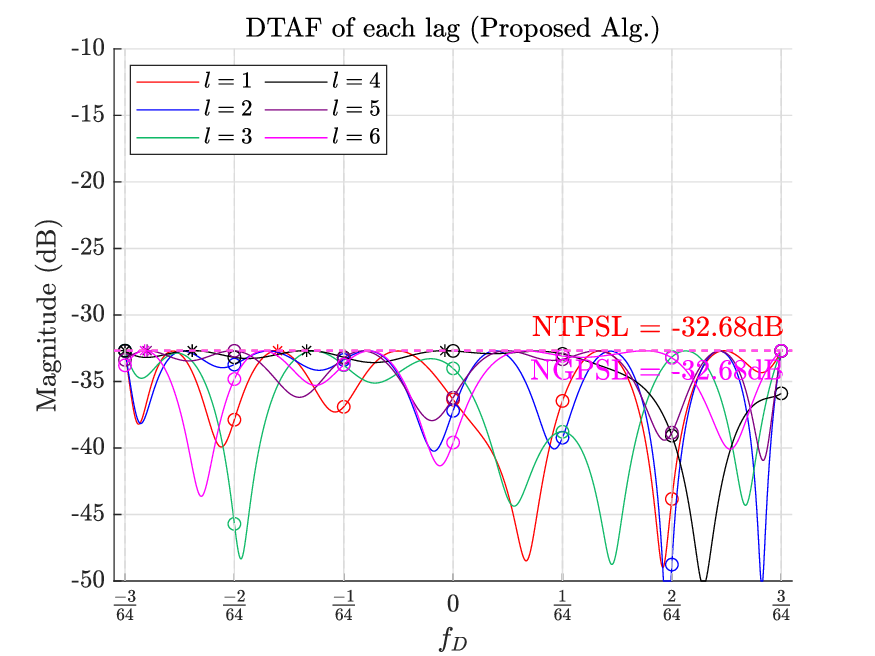}
    \label{Fig:N_64_SROCR_b}} 
    \par \vspace{-0.1mm}  
    \subfloat[]{\includegraphics[width=1.7in]{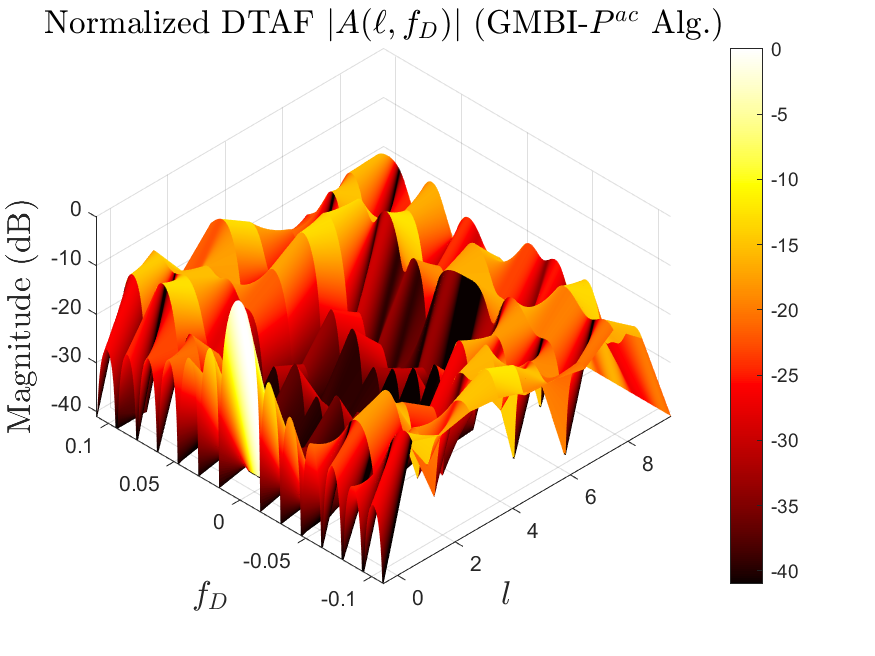}
    \label{Fig:N_64_GMBI_a}} 
    \hspace{-0.5cm} 
    \subfloat[]{\includegraphics[width=1.8in]{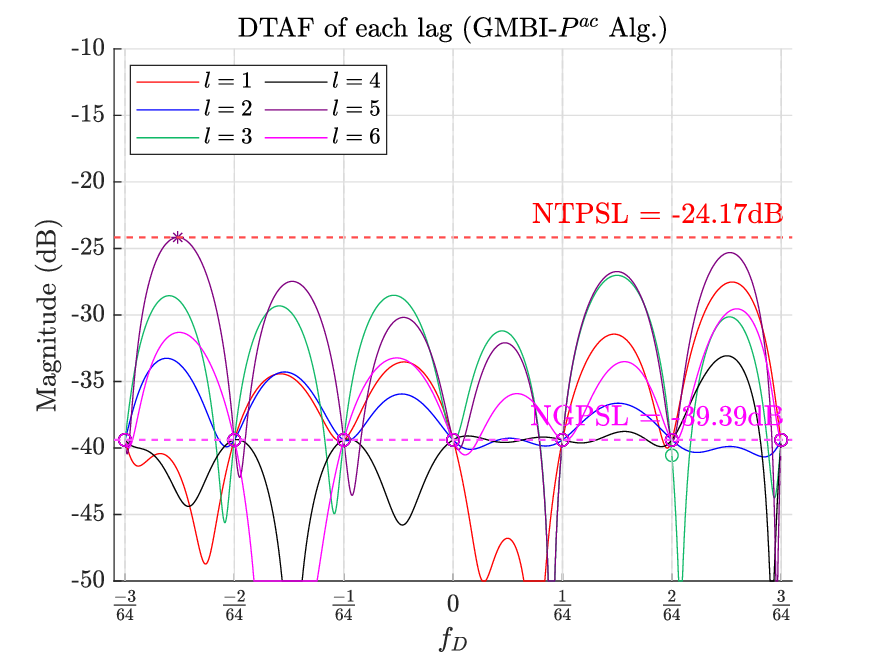}  
    \label{Fig:N_64_GMBI_b}} 
    \par \vspace{-0.1mm}  
    \subfloat[]{\includegraphics[width=1.7in]{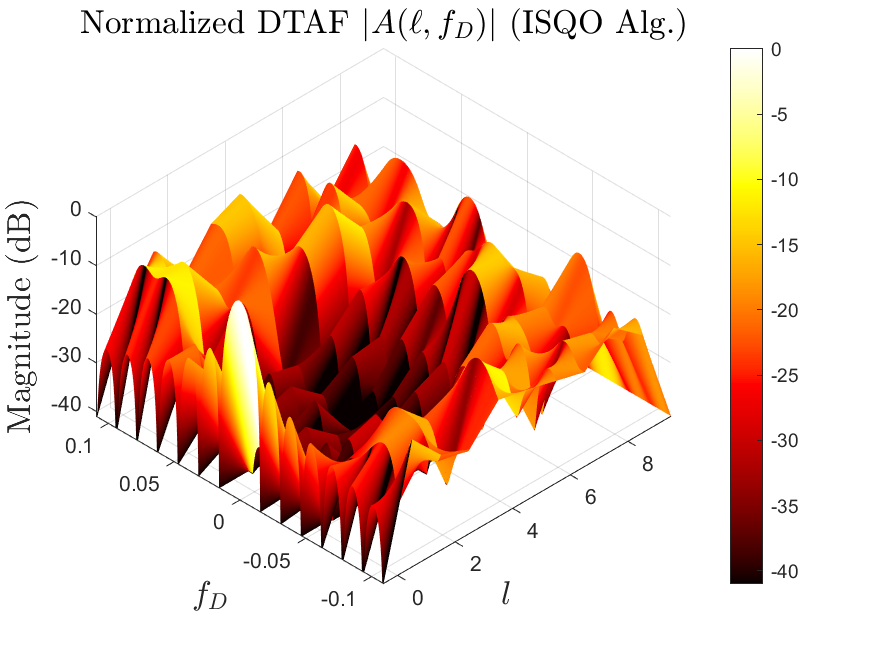}
    \label{Fig:N_64_ISQO_a}} 
    \hspace{-0.5cm} 
    \subfloat[]{\includegraphics[width=1.8in]{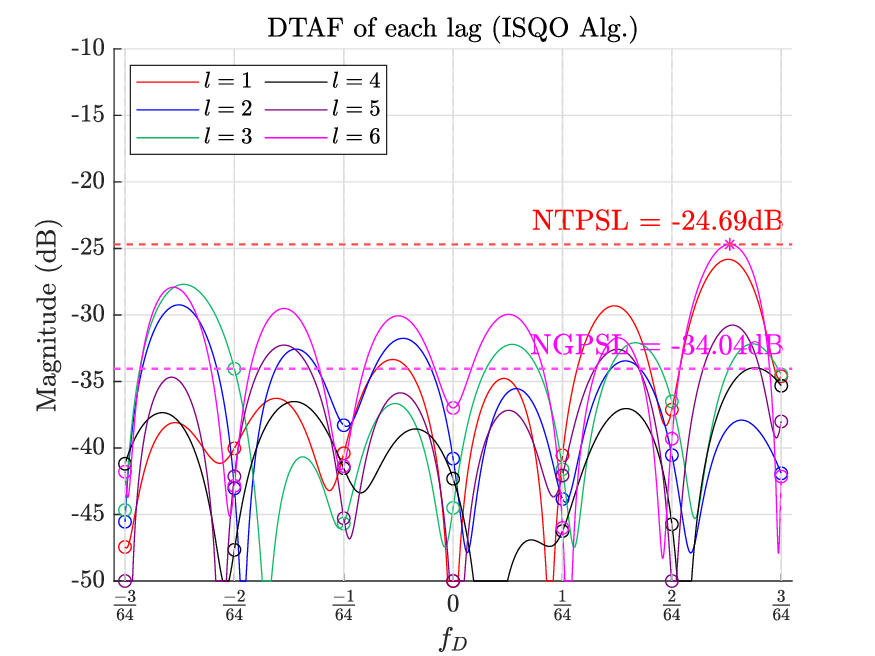}  
    \label{Fig:N_64_ISQO_b}} 
    \caption{ The DTAF and its magnitude of each lag obtained using different methods in Case 2 (\mbox{$N=64$}): (a) and (b) from the proposed algorithm, (c) and (d) from the GMBI-$P^{ac}$ algorithm \cite{Chen22}, (e) and (f) from the ISQO algorithm \cite{Yang18}. }  \label{Fig:N_64} 
\end{figure}
\begin{table}[ht!] 
    \centering
    \caption{NTPSL, NGPSL, and NWISL for different methods in Case 2 (\mbox{$N=64$}).}
    \label{Table:N_64}
    \begin{scriptsize}
    \begin{tabular}{  c  c  c  c } 
        \hline
        \textbf{Method} & \textbf{NTPSL} (${\rm dB}$) & \textbf{NGPSL} (${\rm dB}$) & \textbf{NWISL} (${\rm dB}$)\\
        \hline
        Proposed Algorithm & $-32.68$ & $-32.68$ & $-34.77$ \\
        GMBI-$P^{ac}$ Algorithm \cite{Chen22} & $-24.17$  & $-39.39$ & $-39.42$ \\ 
        ISQO Algorithm \cite{Yang18} & $-24.69$ & $-34.04$ & $-41.28$\\ 
        \hline
    \end{tabular}
    \end{scriptsize}
\end{table}

\subsubsection{Case 3} \label{sec:sim_Case3}
We also design unimodular sequences with a longer length of \mbox{$N = 128$}.
Algorithm \ref{Alg:SROCR-TPSL} is applied to consider the PSL suppression over the sidelobe region
\mbox{$\{ \pm 1, \pm 2,..., \pm 12 \} \times\left[ -\tfrac{3}{128}, \tfrac{3}{128} \right]$}.
In addition, we set {$\zeta = 10$}, {$\kappa=0.99$}, and {$\epsilon=0.005$}.
While, both the GMBI-$P^{ac}$ algorithm and the ISQO algorithm consider the PSL suppression over 
\mbox{$\{ \pm 1,\pm 2,...,\pm 12 \} \times \{ 0, \tfrac{\pm 1}{128}, \tfrac{\pm 2}{128}, \tfrac{\pm 3}{128}\}$}.
Parameters in the GMBI-$P^{ac}$ algorithm are set as \mbox{$J=1$}, \mbox{$\rho=0.01$}, \mbox{$\mu=50$}, \mbox{$T_{\rm mbi}=100$}, and \mbox{$T_0=10000$}.
Additionally, parameters in the ISQO algorithm are set as \mbox{$\beta=1$}, \mbox{$\gamma=1$}, \mbox{$\varepsilon=10^{-6}$}, and \mbox{$w_l=1,\forall l \in \{\pm 1,...,\pm 12\}$}.
In each iteration of Algorithm \ref{Alg:SROCR-TPSL}, 
the values of $\delta^{\langle i \rangle}$, 
$w^{\langle i \rangle}$, and the normalized objective function values of the optimization problem (\ref{P:SROCR_iter}) are shown in Fig.\ref{Fig:N_128_SROCR_parmameters}.
The normalized DTAF results of those sequences are shown in Fig. \ref{Fig:N_128}, and their NTPSL, NGPSL, and NWISL values are listed in Table \ref{Table:N_128}.
From Table \ref{Table:N_128}, we can observe that the proposed algorithm outperforms the GMBI-$P^{ac}$ algorithm and the ISQO algorithm in terms of the NTPSL criterion by \mbox{$14.18 \ {\rm dB}$} and \mbox{$12 \ {\rm dB}$}, respectively.

\begin{figure}[ht!]
    \centering
    \subfloat[]{\includegraphics[width=1.2in]{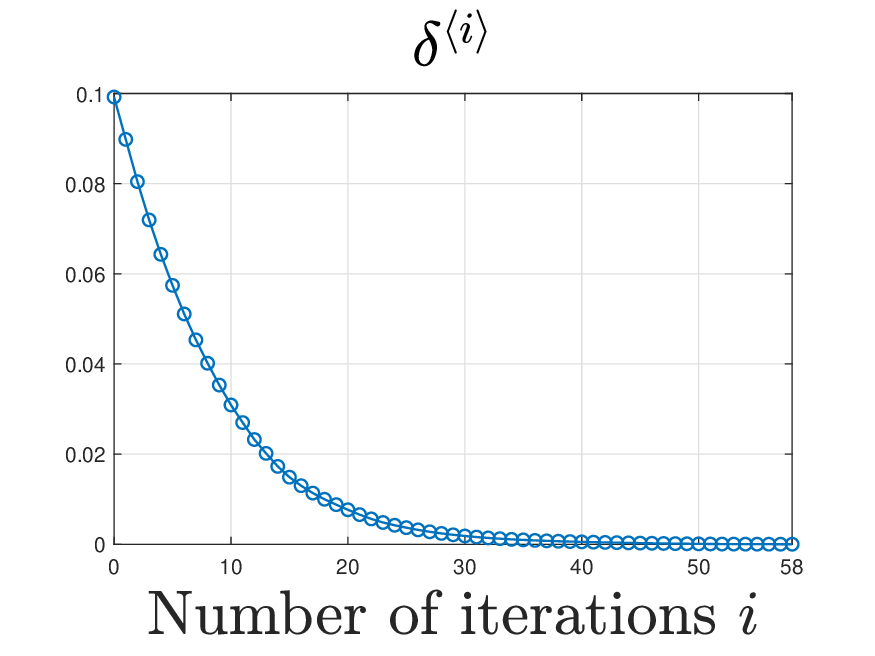} 
    \label{Fig:N_128_SROCR_delta}} 
    \hspace{-0.5cm}
    \subfloat[]{\includegraphics[width=1.2in]{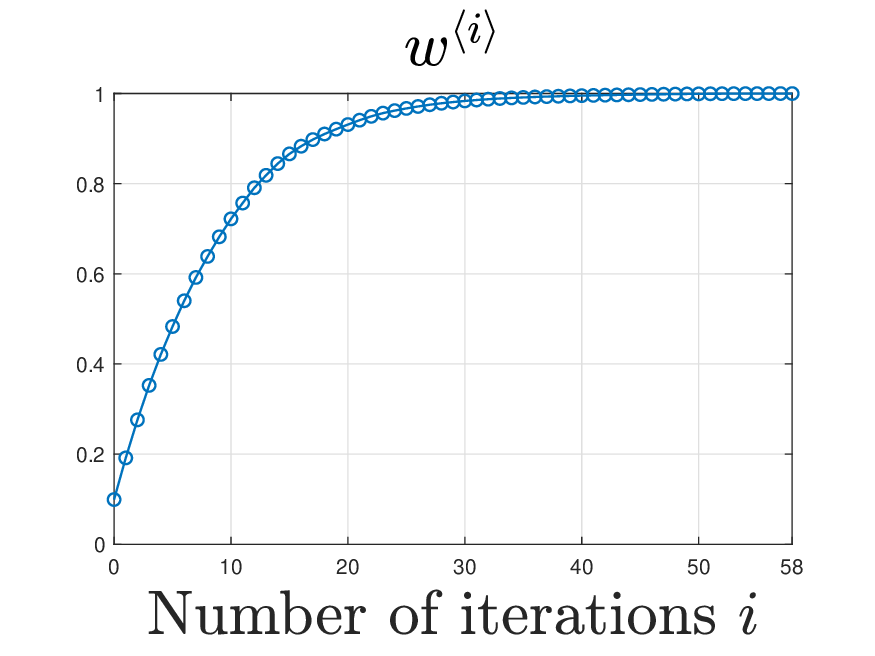}
    \label{Fig:N_128_SROCR_w}} 
    \hspace{-0.5cm}
    \subfloat[]{\includegraphics[width=1.2in]{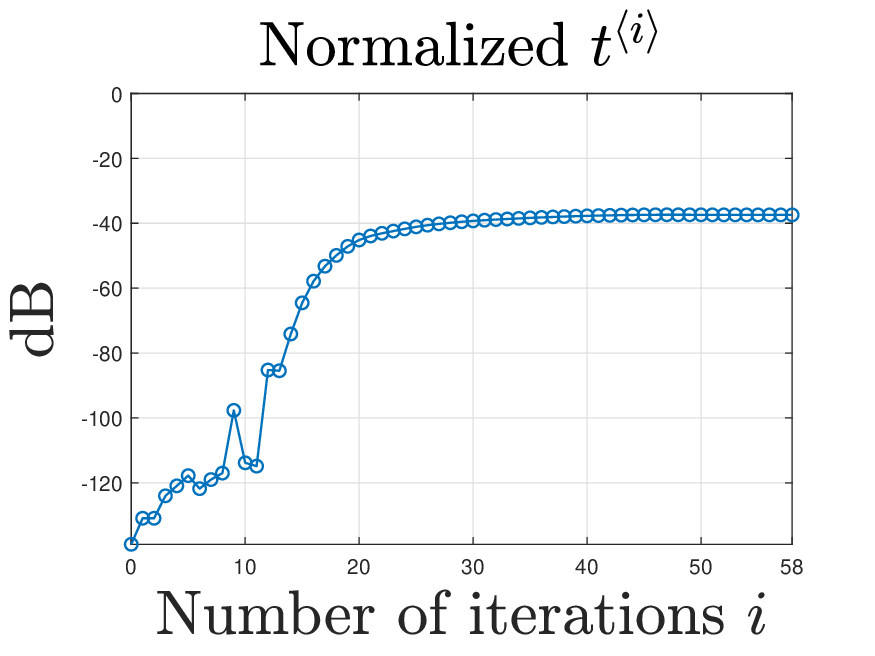} \label{Fig:N_128_SROCR_NTPSL}} 
    \caption{ The parameters of the proposed algorithm at each iteration in Case 1 (\mbox{$N=128$}): (a) $\delta^{\langle i \rangle}$, 
    (b) $w^{\langle i \rangle}$, and (c) normalized $t^{\langle i \rangle}$. }  
    \label{Fig:N_128_SROCR_parmameters}
\end{figure}

\begin{figure}[ht!]
    \centering
    \subfloat[]{\includegraphics[width=1.7in]{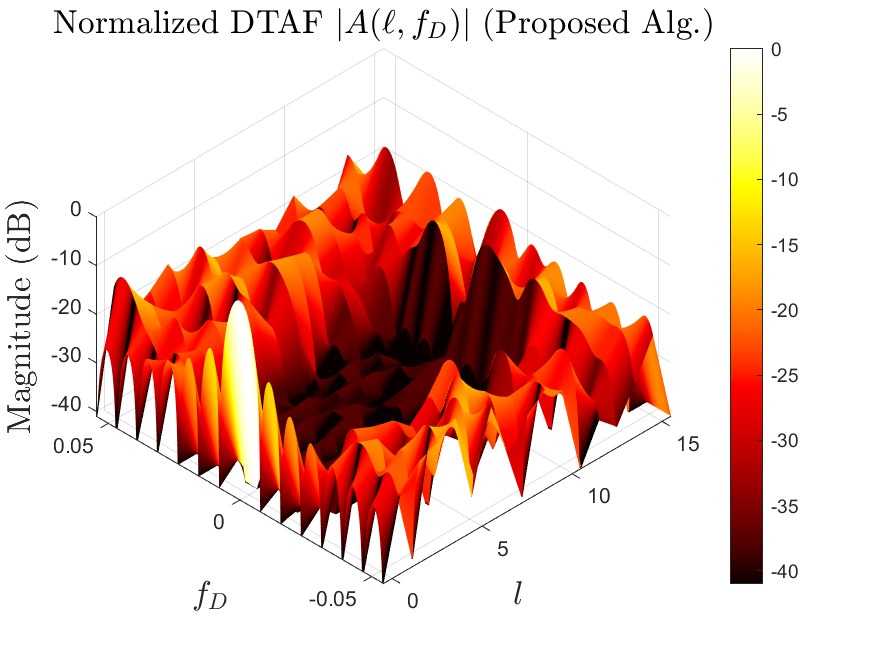} 
    \label{Fig:N_128_SROCR_a}} 
    \hspace{-0.5cm}  
    \subfloat[]{\includegraphics[width=1.8in]{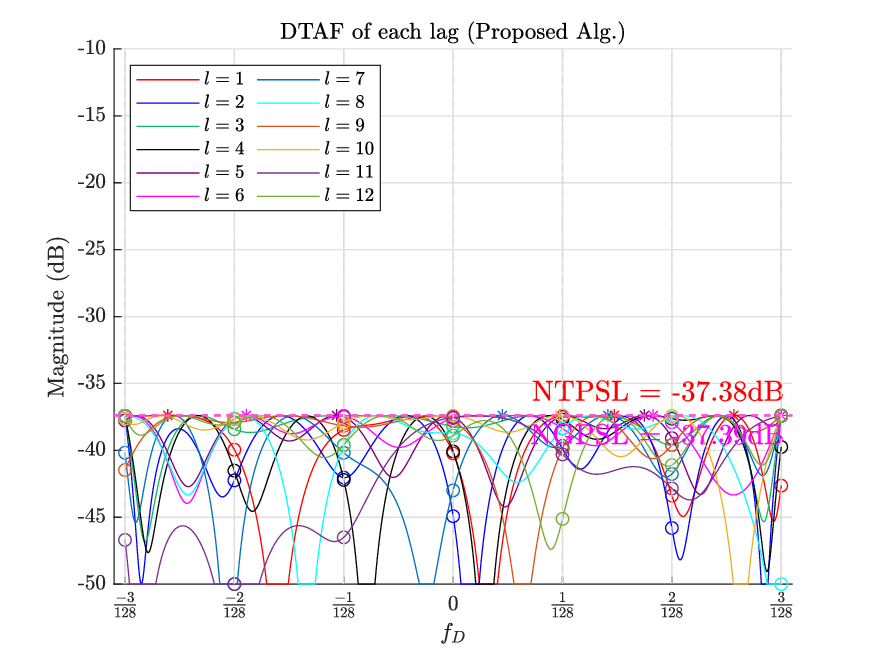}
    \label{Fig:N_128_SROCR_b}} 
    \par \vspace{-0.1mm}  
    \subfloat[]{\includegraphics[width=1.7in]{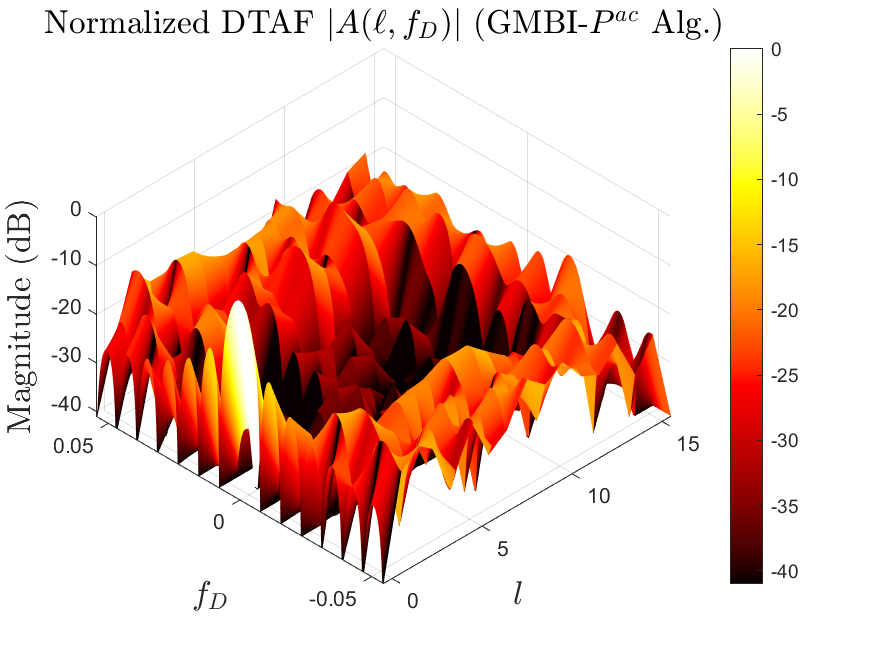}
    \label{Fig:N_128_GMBI_a}} 
    \hspace{-0.5cm} 
    \subfloat[]{\includegraphics[width=1.8in]{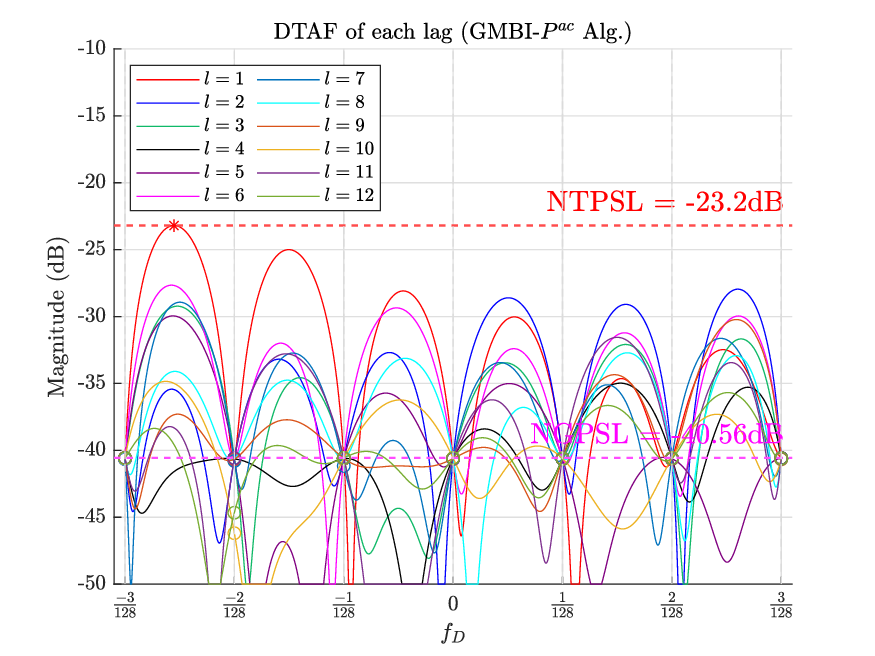}  
    \label{Fig:N_128_GMBI_b}} 
    \par \vspace{-0.1mm}  
    \subfloat[]{\includegraphics[width=1.7in]{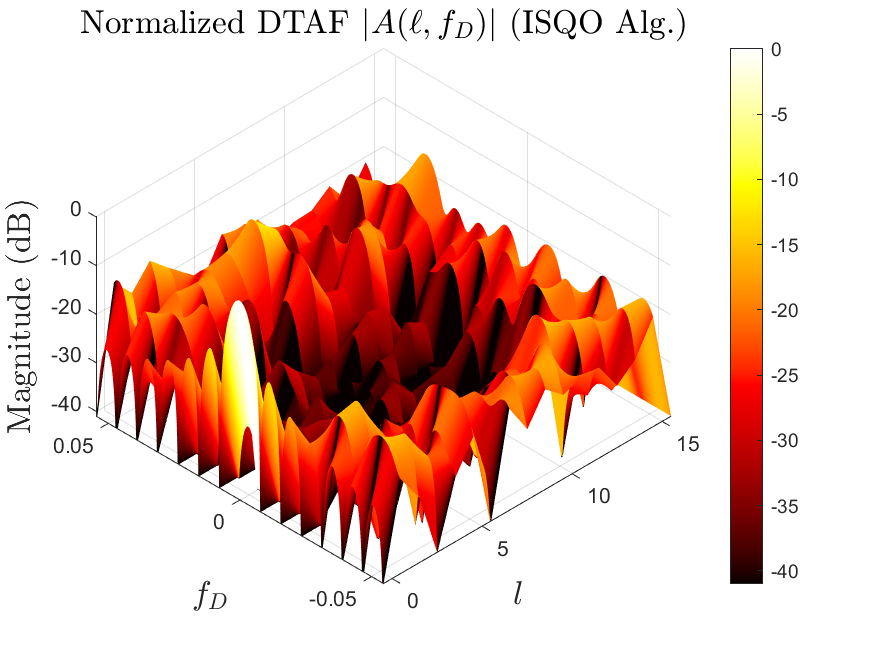}
    \label{Fig:N_128_ISQO_a}} 
    \hspace{-0.5cm} 
    \subfloat[]{\includegraphics[width=1.8in]{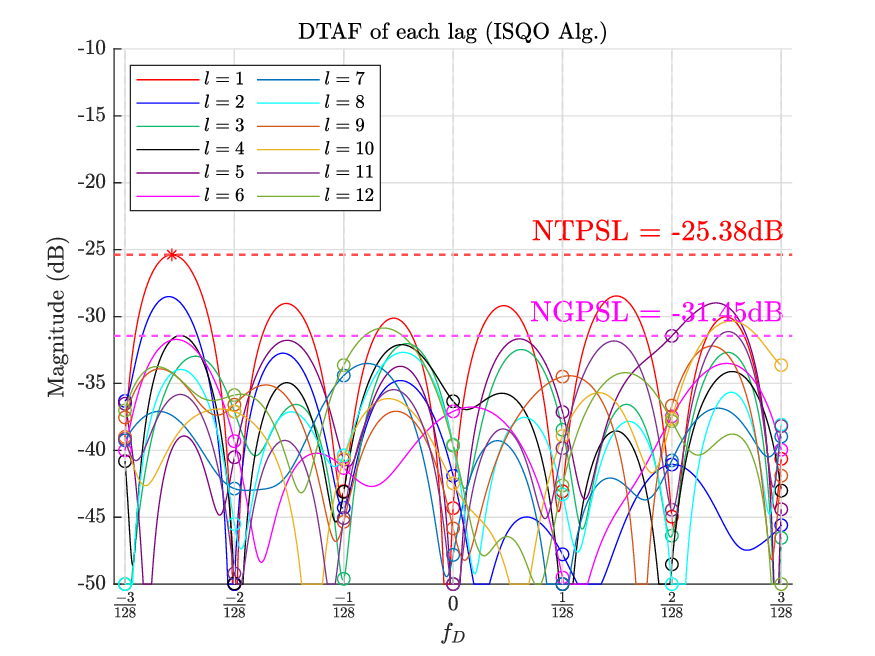}  
    \label{Fig:N_128_ISQO_b}} 
    \caption{ The DTAF and its magnitude of each lag obtained using different methods in Case 3 (\mbox{$N=128$}): (a) and (b) from the proposed algorithm, (c) and (d) from the GMBI-$P^{ac}$ algorithm \cite{Chen22}, (e) and (f) from the ISQO algorithm \cite{Yang18}. }  \label{Fig:N_128}  
\end{figure}

\begin{table}[ht!] 
    \centering
    \caption{NTPSL, NGPSL, and NWISL for different methods in Case 3 (\mbox{$N=128$}).} 
    \label{Table:N_128}
    \begin{scriptsize}
    \begin{tabular}{  c  c  c  c } 
        \hline
        \textbf{Method} & \textbf{NTPSL} (${\rm dB}$) & \textbf{NGPSL} (${\rm dB}$) & \textbf{NWISL} (${\rm dB}$) \\
        \hline
        Proposed Algorithm & $-37.38$ & $-37.39$ & $-39.20$ \\
        GMBI-$P^{ac}$ Algorithm \cite{Chen22} & $-23.20$ & $-40.56$ & $-40.71$ \\ 
        ISQO Algorithm \cite{Yang18} & $-25.38$ & $-31.45$ & $-40.79$ \\ 
        \hline
    \end{tabular}
    \end{scriptsize}
\end{table}

The numerical results in Case 1 (\mbox{$N=32$}), Case 2 (\mbox{$N=64$}), and Case 3 (\mbox{$N=128$}) demonstrate that the proposed algorithm outperforms the GMBI-$P^{ac}$ algorithm and the ISQO algorithm in achieving a lower NTPSL.
The primary reason for this advantage of the proposed algorithm is its ability to suppress PSL of the DTAF across a continuous Doppler frequency shift ROI.
However, the GMBI-$P^{ac}$ algorithm and the ISQO algorithm optimize the waveform by considering only the Doppler frequency shift bins.
Therefore, the PSL suppression over the continuous Doppler frequency shift ROI is not guaranteed. 
It can be observed from Table \ref{Table:N_32}, \ref{Table:N_64}, and \ref{Table:N_128} that a lower NGPSL or NWISL does not imply a lower NTPSL.
Moreover, the PSL can be misjudged when using the NGPSL or NWISL criterion instead of the NTPSL criterion. 
In Section \ref{sec:Application}, we will demonstrate that sidelobes of AF should be suppressed over a continuous Doppler frequency shift ROI to prevent false alarms in target detection for arbitrary speeds.

\subsubsection{Consider various $\zeta$ settings in \eqref{delta_i} of the proposed algorithm} \label{subsubsec:zeta_settings}
In the following, we use Section \ref{sec:Numerical_Ex} Case 1 (\mbox{$N=32$}) as an example to show that different $\zeta$ settings in \eqref{w_0} and \eqref{delta_i} of the proposed algorithm may affect both the NTPSL performance and the number of iterations required for convergence.
The results corresponding to different $\zeta$ settings are presented in Table \ref{Table:N_32_zeta}.
Our experiments indicate that a smaller $\zeta$ value typically leads to the proposed algorithm converging in fewer iterations.
This is mainly because a smaller $\zeta$ leads to a larger increment of $w^{\langle i \rangle}$ in \eqref{w_i}, which accelerates the approximation of the constraint in \eqref{P:SROCR_iter_f} toward the rank-one constraint in \eqref{opt: rank-1 constraint of problem during SROCR}.
From Table \ref{Table:N_32_zeta}, we can observe that \mbox{$\zeta=50$} results in the lowest NTPSL (\mbox{$-29.52 \ {\rm dB}$}) taking $342$ iterations for the proposed algorithm to converge.
Meanwhile, \mbox{$\zeta=10$} results in a slightly higher NTPSL (\mbox{$-29.30\ {\rm dB}$}) but requires only $81$ iterations, which is about $4$ times fewer than the number of iterations required when \mbox{$\zeta=50$}.
Considering both the NTPSL performance and convergence efficiency, we select \mbox{$\zeta=10$} as a representative setting to present the results in Section \ref{sec:Numerical_Ex} Case 1 (\mbox{$N=32$}).

\begin{table}[ht!] 
    \centering
    \caption{Results for various $\zeta$ settings of the proposed algorithm in Case 1 (\mbox{$N=32$}).}
    \label{Table:N_32_zeta}
    \begin{scriptsize}
    \begin{tabular}{  c  c  c  c  c  } 
        \hline
        \textbf{$\zeta$} & \textbf{NTPSL} (${\rm dB}$) & \textbf{NGPSL} (${\rm dB}$) & \textbf{NWISL} (${\rm dB}$) & \textbf{Iterations} \\
        \hline
        \hline
         $2$  & $-24.53$ & $-24.53$ & $-27.58$ & $9$   \\
        \hline
         $5$  & $-29.22$ & $-29.22$ & $-31.91$ & $29$ \\
        \hline
         $10$ & $-29.30$ & $-29.30$ & $-31.25$ & $81$ \\
        \hline
         $20$ & $-28.90$ & $-28.92$ & $-31.42$ & $139$ \\
        \hline
         $30$ & $-28.75$ & $-28.78$ & $-31.99$ & $178$ \\
        \hline
         $40$ & $-29.42$ & $-29.42$ & $-31.39$ & $286$\\
        \hline
         $50$ & $-29.52$ & $-29.54$ & $-31.93$ & $342$ \\
        \hline
         $60$ & $-29.34$ & $-29.34$ & $-30.75$ & $954$ \\
        \hline
    \end{tabular}
    \end{scriptsize}
\end{table}

\subsection{Application in target detection} \label{sec:Application}
In the following, the waveforms generated in Section \ref{sec:Numerical_Ex} Case 3 (\mbox{$N=128$}) are applied to target detection scenarios to demonstrate that the waveform designed by the proposed algorithm is less likely to cause false alarms, particularly when the target speed does not correspond to an integer multiple of the system's velocity resolution.
We consider a monostatic radar system in which the operating wavelength is \mbox{$\lambda=0.02$} m, and the sampling frequency is \mbox{$f_s=1 \ {\rm MHz}$}.
The velocity resolution of the system can be obtained as \mbox{$\tfrac{\lambda f_s}{2N}=78.13 \ {\rm m/s}$} \cite{Levanon04}. 
In addition, the range resolution of the system can be calculated as \mbox{$\tfrac{cT_s}{2}=150 \ {\rm m}$} \cite{Levanon04}, where $c$ is the speed of light and \mbox{$T_s=\tfrac{1}{f_s}$}. 
Moreover, suppose that a false alarm probability $P_{\rm FA}$ is to be achieved, then the detection threshold can be set as \mbox{$\eta=-\ln(P_{\rm FA})\sigma_w^2$} \cite{Richards14} 
\footnote{
    According to \cite{Richards14} [Sec 6.3.2], the detection threshold \mbox{$\eta=-\ln(P_{\rm FA})\sigma_w^2$} is derived based on the detection of a nonfluctuating target in white Gaussian noise with a square-law detector.
},
where $\sigma_w^2$ is the noise variance.
The noise is assumed to be a zero-mean white Gaussian noise with a variance of \mbox{$-45 \ {\rm dB}$}.
We choose \mbox{$P_{\rm FA}=10^{-8}$}, and then the detection threshold can be calculated as \mbox{$\eta=-\ln(10^{-8})10^{\tfrac{-45}{10}}$} (\mbox{$=-32.35 \ {\rm dB}$}).

We first assume that a stationary target is located at \mbox{$22,200 \ {\rm m}$}, and the power of the received echo signal reflected from the target is \mbox{$0 \ {\rm dB}$}.
In Figs. \ref{Fig:scenario_1_proposed_alg}, \ref{Fig:scenario_1_GMBI}, and \ref{Fig:scenario_1_ISQO}, range-velocity plots computed from the received echo signals using the transmitted waveforms generated in Section \ref{sec:Numerical_Ex} Case 3 (\mbox{$N=128$}) are illustrated.
Since the echo signal reflected from the target exceeds the detection threshold in these plots, the target is successfully detected. 
In addition, the sidelobes in Figs. \ref{Fig:scenario_1_proposed_alg} and \ref{Fig:scenario_1_GMBI} are lower than the detection threshold \mbox{$\eta=-32.25 \ {\rm dB}$}, and therefore do not result in false alarms.
Only sporadic peaks exceeding the threshold $\eta$ are found in Fig. \ref{Fig:scenario_1_ISQO} (marked with ``x''). 

However, when the target speed does not correspond to an integer multiple of the system's velocity resolution, severe false alarms occur when using the transmitted waveforms generated by the GMBI-$P^{ac}$ algorithm \cite{Chen22} and the ISQO algorithm \cite{Yang18}.
Assume that the target is moving toward the radar with a speed of $-26.04$ m/s (roughly one-third of the system's velocity resolution).
In Figs. \ref{Fig:scenario_2_GMBI} and \ref{Fig:scenario_2_ISQO},
severe false alarms occur and are marked. 
In contrast, no false alarms occur in Fig. \ref{Fig:scenario_2_proposed_alg} when using the proposed waveform.
The primary reason for the occurrence of false alarms in Figs. \ref{Fig:scenario_2_GMBI} and \ref{Fig:scenario_2_ISQO} is that the GMBI-$P^{ac}$ algorithm and the ISQO algorithm only consider PSL or WISL suppression on an integer multiple of Doppler frequency shift bins.
They can ensure low sidelobe levels in the range-velocity plots when the target speed corresponds to an integer multiple of the system's velocity resolution, as displayed in \cite{Chen22} and \cite{Yang18}, and our demonstrations in Figs. \ref{Fig:scenario_1_GMBI} and \ref{Fig:scenario_1_ISQO}.
However, when the target speed is a fractional multiple of the system's velocity resolution, the unexpected sidelobes that appear in the range-velocity plots may cause false alarms.
Moreover, we can observe in Figs \ref{Fig:scenario_2_proposed_alg}, \ref{Fig:scenario_2_GMBI}, and \ref{Fig:scenario_2_ISQO} that a bright strip appears at \mbox{$22,200 \ {\rm m}$}. This occurs because the zero-delay cut of the DTAF for the unimodular sequence is approximated to a sinc function \cite{Zhang16} and cannot be optimized. Despite this, post-processing techniques can be applied to mitigate the impact of the sinc function during target detection.
In summary, the proposed algorithm considering PSL suppression over a continuous Doppler frequency shift ROI is of practical importance to prevent false alarms for arbitrary target velocities.

\begin{figure}[ht!]
    \centering
    \subfloat[]{\includegraphics[width=1.7in]{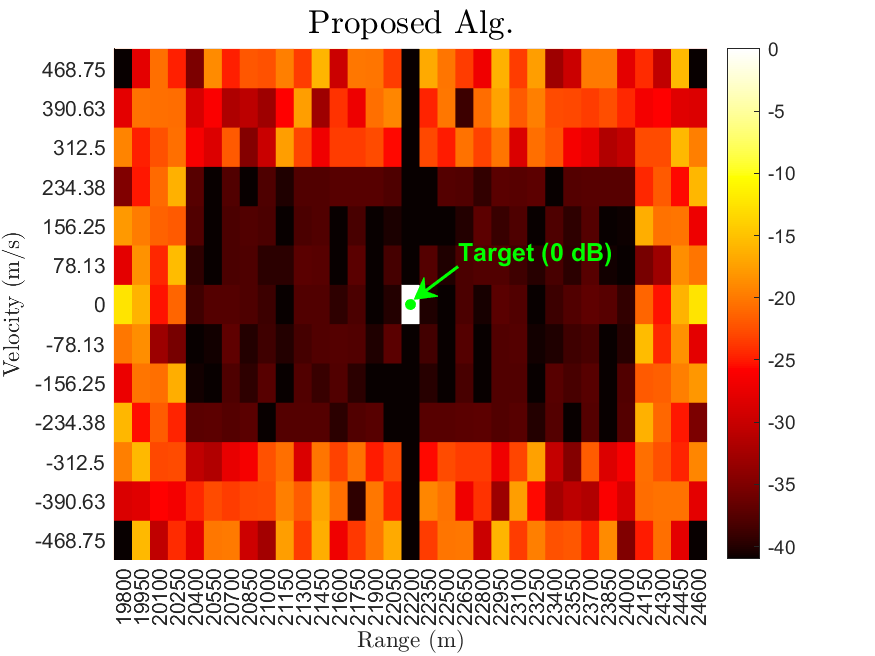} 
    \label{Fig:scenario_1_proposed_alg}} 
    \hspace{-0.5cm}  
    \subfloat[]{\includegraphics[width=1.7in]{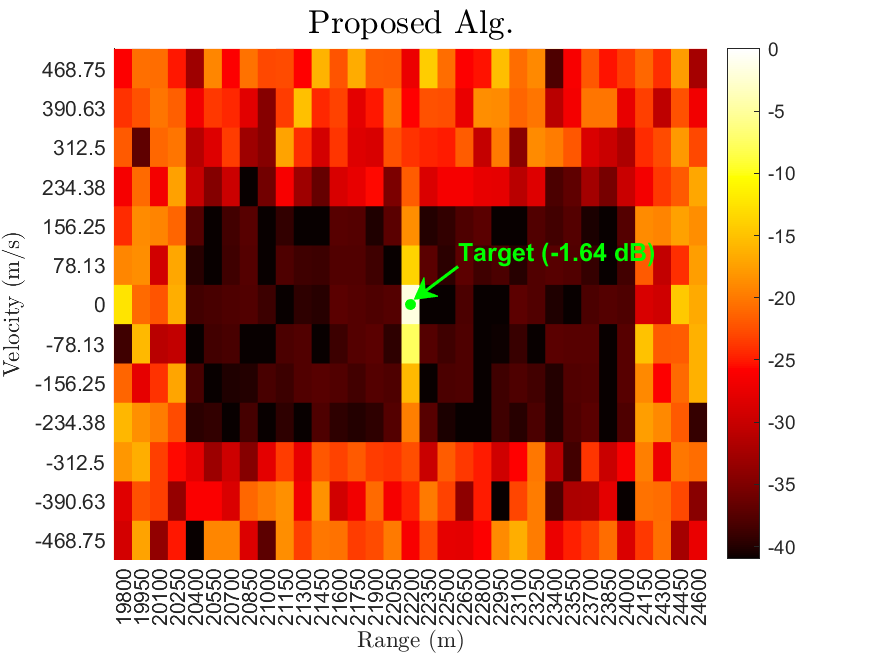}
    \label{Fig:scenario_2_proposed_alg}} 
    \par \vspace{-0.5mm}  
    \subfloat[]{\includegraphics[width=1.7in]{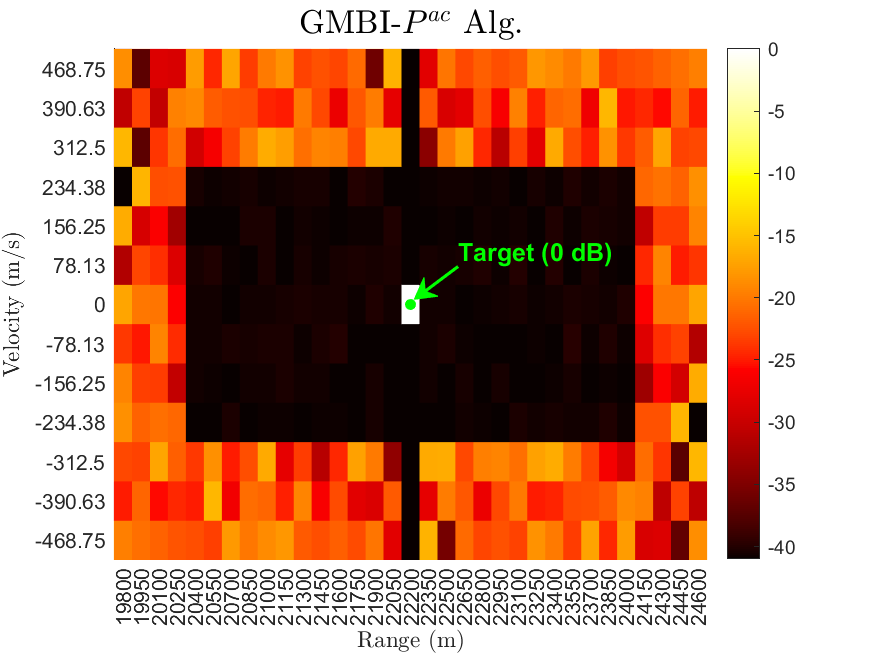}
    \label{Fig:scenario_1_GMBI}} 
    \hspace{-0.5cm} 
    \subfloat[]{\includegraphics[width=1.7in]{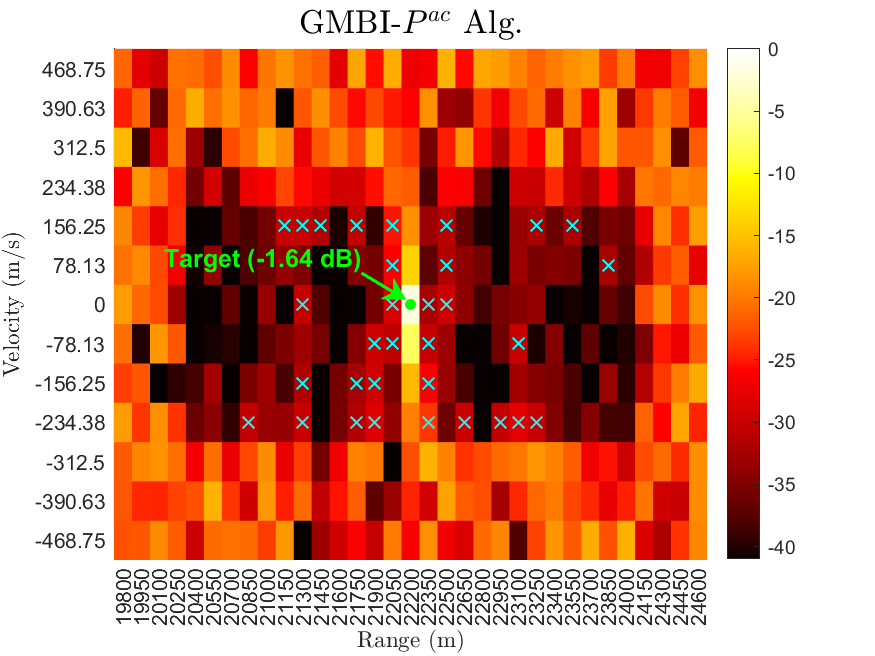}  
    \label{Fig:scenario_2_GMBI}} 
    \par \vspace{-0.5mm}  
    \subfloat[]{\includegraphics[width=1.7in]{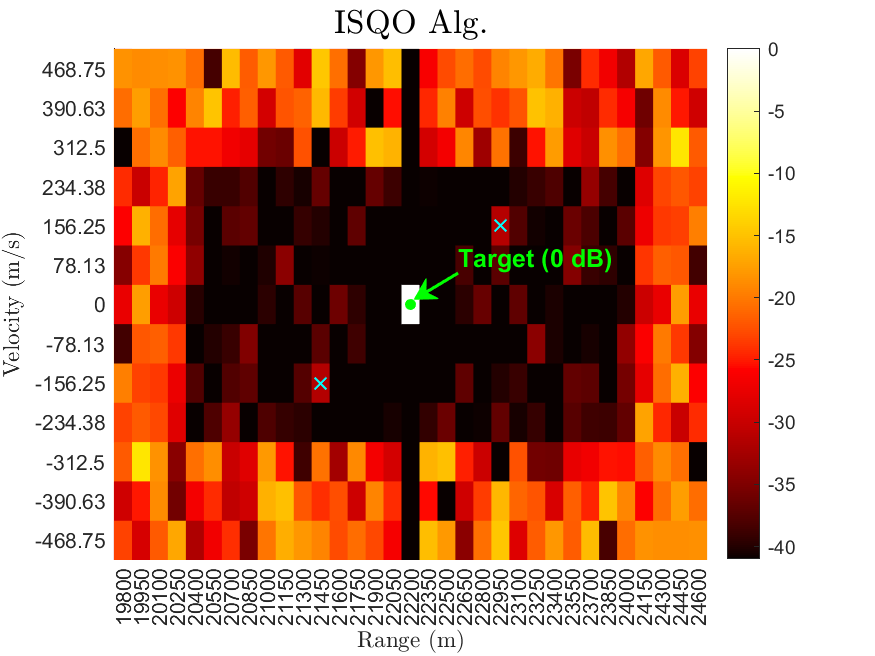}
    \label{Fig:scenario_1_ISQO}} 
    \hspace{-0.5cm} 
    \subfloat[]{\includegraphics[width=1.7in]{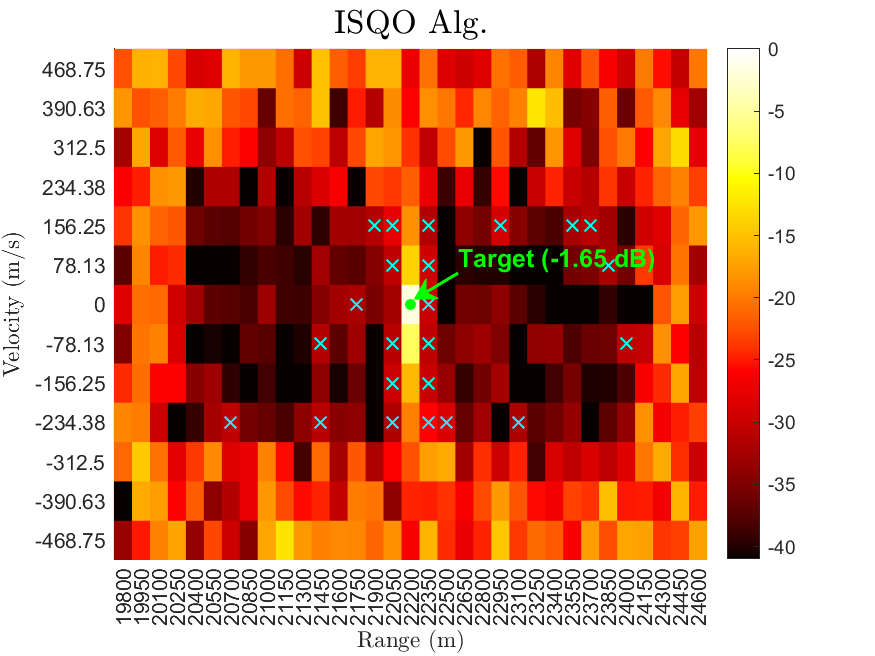}  
    \label{Fig:scenario_2_ISQO}} 
    \vspace{-0.1mm}
    \caption{Range-velocity plots for target detection scenario using the waveforms generating in Section \ref{sec:Numerical_Ex} Case 3 (\mbox{$N=128$}): (a) and (b) from the proposed algorithm, (c) and (d) from the GMBI-$P^{ac}$ algorithm \cite{Chen22}, (e) and (f) from the ISQO algorithm \cite{Yang18} (The occurrence of false alarms are marked with ``x").} \label{Fig:Range-velocity_images}
\end{figure}

\section{Conclusion} \label{sec:conclusion}
In this paper, a new method is proposed to design unimodular waveforms based on the discrete-time ambiguity function (DTAF) optimization problem, in which the peak sidelobe level (PSL) is suppressed over a continuous Doppler frequency shift region of interest (ROI).
Unlike previous works that suppress PSL only at grid points in the delay-Doppler plane and exhibit low false alarms only when the target speed is an integer multiple of the system's velocity resolution, the waveform designed by the proposed method can prevent false alarms for detecting a target with an arbitrary velocity.
The proposed method tackles the non-convex semi-infinite programming (SIP) by showing that infinite constraints over a continuous Doppler frequency shift region are equivalent to a finite number of linear matrix inequalities (LMIs).
The resulting semidefinite programming (SDP) with a rank-one constraint is then solved by the sequential rank-one constraint relaxation (SROCR) algorithm.
Numerical results demonstrate that the proposed method outperforms existing methods by a considerable margin in achieving a lower normalized true peak sidelobe level (NTPSL), with PSL evaluated over a continuous Doppler frequency shift ROI.
Moreover, the waveform designed by the proposed algorithm effectively prevents false alarms when detecting a target moving at arbitrary speeds.
In the future, it would be valuable to establish optimality conditions for the proposed method, such as deriving the dual problem.
\appendices
\section{Derivation of \eqref{eq:CTAF_DTAF}}  \label{appendix:CTAF_DTAF_deriv}

Firstly, by substituting $x(t)$, as defined in (\ref{eq:x(t)}), into (\ref{def:CTAF}), we obtain (\ref{eq:CTAF_DTAF_deriv_d}). 
Letting $l=n-m$, we have that \eqref{eq:CTAF_DTAF_deriv_d} becomes \eqref{eq:CTAF_DTAF_deriv_e}. 
Moreover, letting $s=t-nT_s$, we obtain \eqref{eq:CTAF_DTAF_deriv_f}. 
Lastly, \eqref{eq:CTAF_DTAF_deriv_h} and \eqref{eq:CTAF_DTAF_deriv_i} are derived, where $A(l,f)$ and $\chi(\tau,f)$ are defined in \eqref{def:A(l,f)} and \eqref{def:chi}, respectively.
Thus, the derivation of \eqref{eq:CTAF_DTAF} is completed.

\begin{footnotesize}
\begin{figure*}[ht!]  
    \centering
    \begin{subequations} \label{eq:CTAF_DTAF_deriv}
    \begin{align}
        \hspace{-8mm}
        A_c(\tau,f) 
        &=\sum_{n=0}^{N-1}\sum_{m=0}^{N-1}x_nx_m^*\int_{-\infty}^{\infty}p(t-nT_s)p(t-mT_s-\tau)e^{-j2\pi f(t-\tau)}dt \label{eq:CTAF_DTAF_deriv_d}\\
        &=\sum_{l=-N+1}^{N-1}\sum_{n={\max}(0,l)}^{{\min}(N-1,N-1+l)}x_nx_{n-l}^*\int_{-\infty}^{\infty}p(t-nT_s)p(t-(n-l)T_s-\tau)e^{-j2\pi f(t-\tau)}dt \label{eq:CTAF_DTAF_deriv_e}\\
        &=\sum_{l=-N+1}^{N-1}\sum_{n={\max}(0,l)}^{{\min}(N-1,N-1+l)}x_nx_{n-l}^*\int_{-\infty}^{\infty}p(s)p(s+lT_s-\tau)e^{-j2\pi f(s + nT_s-\tau)}ds \label{eq:CTAF_DTAF_deriv_f}\\
        &=\sum_{l=-N+1}^{N-1}\sum_{n={\max}(0,l)}^{{\min}(N-1,N-1+l)}x_nx_{n-l}^*e^{-j2\pi fT_s(n-l)}\int_{-\infty}^{\infty}p(s)p(s+lT_s-\tau)e^{-j2\pi f(s + lT_s-\tau)}ds \label{eq:CTAF_DTAF_deriv_h}\\
        &=\sum_{l=-N+1}^{N-1}A(l,f)\chi(\tau-lT_s,f). \label{eq:CTAF_DTAF_deriv_i}
    \end{align}
    \end{subequations}  
\end{figure*}
\end{footnotesize}

\section{Proof of Theorem \ref{thm:SIP_LMI}} \label{appendix:proof_thm_SIP_LMI}

To prove Theorem \ref{thm:SIP_LMI}, we first introduce the following lemma, which will be applied to deal with the infinite constraints in (\ref{P2_2_b}) later.

\begin{lemma} (A special case of Theorem 3 in \cite{Dumitrescu2006} and Corollary 4.27 in \cite{Dumitrescu2017}, considering a univariate trigonometric polynomial) \label{lemma:BRL}
    Define $H(f)$ as
    \begin{equation} \label{eq:causal_trig}
        H(f) = \sum_{n = 0}^{K} h_ne^{-j2\pi fn}, \quad \forall f \in \left[-\tfrac{1}{2}, \tfrac{1}{2}\right].
    \end{equation}
    Given a positive real number $\gamma$, the following inequality holds
    \begin{equation} \label{BRL_inequality}
        |H(f)| \leq \gamma,\ \forall f \in [\alpha, \beta] \subset \left[-\tfrac{1}{2}, \tfrac{1}{2}\right]
    \end{equation}
    if and only if there exist matrices ${\bf Q} \in \mathbb{H}_+^{K+1}$ and ${\bf P} \in \mathbb{H}_+^{K}$ such that
    \begin{subequations} \label{eq:BRL_LMI}
    \begin{align} 
        & \gamma^2 \delta_n = \mathrm{Tr}\left( \boldsymbol{\Theta}_{K+1}^{(n)}{\bf Q} \right) 
        + \mathrm{Tr}\Big({\bf{\Phi}}_{K}^{(n)} {\bf P} \Big), \ n \in \mathbb{Z}_{K+1}, \label{eq:BRL_LMI_a} \\ 
        \renewcommand{\arraystretch}{0.8}
        &\left[\begin{array}{cc}
            \mathbf{Q} & \mathbf{h} \\
            \mathbf{h}^H & 1
        \end{array}\right] \succeq \mathbf{O}_{K+2},  \label{eq:BRL_LMI_b}
    \end{align}
    \end{subequations}
    where
    ${\bf \Theta}_{K}^{(n)}$ is defined in \eqref{def:Elementary_Toeplitz_matrix},
    ${\delta}_n$ is defined in \eqref{def:delta},
    ${\bf{\Phi}}_{K}^{(n)}$ is defined in \eqref{def:Phi},
    and $\mathbf{h} = [h_0, h_1, \ldots, h_{N-1}]^T$.

    It is worth noting that Theorem 3 in \cite{Dumitrescu2006} and Corollary 4.27 in \cite{Dumitrescu2017} deal with a multivariate trigonometric polynomial, while Lemma \ref{lemma:BRL} specifically considers a univariate trigonometric polynomial as $H(f)$ defined in \eqref{eq:causal_trig}.
\end{lemma}

\begin{proof}
    The proof of Lemma \ref{lemma:BRL}, based on the proof idea of \cite{Dumitrescu2006} and \cite{Dumitrescu2006}, is provided in Appendix \ref{appendix:proof_thm_BRL} using the notation of this paper for clarity.
\end{proof}

To apply Lemma \ref{lemma:BRL} for tackling the infinite constraints in (\ref{P2_2_b}),
we further define a polynomial
\begin{align} \label{eq: delay-product signal definition}
    H^{(l)}(f_D) &= \sum_{n={\max}(0,l)}^{{\min}(N-1,N-1+l)}h_n^{(l)}e^{-j2\pi {f_D}n}, \\
    & \forall l \in \mathbb{Z}_N,\ \forall f_D \in [-f_R, f_R] \subseteq \big[\tfrac{-1}{2},\tfrac{1}{2}\big], \nonumber 
\end{align}
where the coefficients of $H^{(l)}(f_D)$ are $h^{(l)}_n = x_nx_{n-l}^*$, $\forall n \in \mathbb{Z}_N$, $\forall l \in \mathbb{Z}_N$. 
Then, it can be verified that $| H^{(l)}(f_D) | = | A(l  ,f_D) |$:

\vspace{-2mm}
\begin{small}
\begin{subequations}
\begin{align}
    | H^{(l)}(f_D) |
    &= \left| \sum_{n={\max}(0,l)}^{{\min}(N-1,N-1+l)}h_n^{(l)}e^{-j2\pi {f_D}n} \right| \nonumber \\
    &= \left| \sum_{n={\max}(0,l)}^{{\min}(N-1,N-1+l)}x_nx_{n-l}^*e^{-j2\pi {f_D}n} \right| \nonumber \\
    &= \left| \sum_{n={\max}(0,l)}^{{\min}(N-1,N-1+l)}x_nx_{n-l}^*e^{-j2\pi {f_D}(n-l)} \right|  
    = | A(l,f_D) |. \nonumber
\end{align}
\end{subequations}
\end{small}

Introduce the auxiliary variables ${\bf h}^{(l)} = \left[ h_0^{(l)}, \ldots, h_{N-1}^{(l)} \right]^T$, $\forall l \in \{1,...,L\}$.
The optimization problem (\ref{P2_2}) can be reformulated as
\begin{subequations} \label{P3}
\begin{align} 
    &\underset{\substack{t \in \mathbb{R}, \ {\bf x} \in {\mathbb C}^N \\{\bf h}^{(1)}, \ldots, {\bf h}^{(L)} \in {\mathbb C}^N}}{\mathrm{minimize}} \quad t \label{P3_a}\\ 
    &\hspace{-1mm} \mathrm{subject\ to} \hspace*{1mm} \left|H^{(l)}(f_D)\right|^2 \leq t, \nonumber \\
    & \hspace*{28mm} \forall l \in \{1,...,L\},\ \forall f_D \in [-f_R,f_R] \label{P3_b}\\
    &\hspace*{16mm} h_n^{(l)} = x_nx_{n-l}^*,\ \forall n \in \mathbb{Z}_N,\ \forall l \in \{1,...,L\} \label{P3_c}\\
    &\hspace*{16mm} |x_n| = 1,\ \forall n \in \mathbb{Z}_N. \label{P3_d}
\end{align}
\end{subequations}

With respect to each time delay bin $l$, \mbox{$\forall l \in \{1,...,L\}$}, Lemma \ref{lemma:BRL} is applied to tackle (\ref{P3_b}) by setting \mbox{$K=N-1$} in (\ref{eq:causal_trig}), \mbox{$\gamma = \sqrt{t}$} and \mbox{$\beta = -\alpha = 2\pi f_R$} in (\ref{BRL_inequality}). 
Further introducing variables \mbox{${\bf Q}^{(1)}, \ldots, {\bf Q}^{(L)} \in {\mathbb H}_+^{N}$} and \mbox{${\bf P}^{(1)}, \ldots, {\bf P}^{(L)} \in {\mathbb H}_+^{N-1}$},
and substituting \eqref{P3_b} with \eqref{eq:BRL_LMI_a} and \eqref{eq:BRL_LMI_b}, we then obtain the optimization problem \eqref{P_LMI_0} and prove that the problem \eqref{P2_2} is equivalent to \eqref{P_LMI_0}.

\section{Proof of Lemma \ref{lemma:BRL}} \label{appendix:proof_thm_BRL}

Before proceeding with the proof of Lemma \ref{lemma:BRL}, we introduce some notation to facilitate the later derivation.
Define $\mathbb{C}[z]$ as the set of trigonometric polynomials with complex coefficient:
$\mathbb{C}[z] = \{ {\breve R}(z)=\sum_{n=-K}^{K}r_nz^{-n}, r_{-n}=r_n^* | K\in\mathbb{N}, \ r_{-K},..., r_K \in \mathbb{C} \}$,
where \mbox{$\mathbb{N}=\{0,1,2,3,...\}$} is the set of natural numbers.
The degree of ${\breve R}\in\mathbb{C}[z]$ is defined as 
${\rm deg} \ {\breve R} = \max\{n \ | \ r_n\neq 0, n = -K, -K+1, ...,K\}$.
Additionally, on the unit circle (i.e., \mbox{$z=e^{j2\pi f}$}, $\forall f \in [\frac{-1}{2},\frac{1}{2}]$), we define \mbox{$R(f)={\breve R}(e^{j2\pi f})=\sum_{n=-K}^{K}r_ne^{-j2\pi f n}$}.
Moreover, the causal polynomial is denoted as 
\begin{equation} \label{eq:H}
    {\breve H}(z)=\sum_{n=0}^{K}h_nz^{-n}, \ h_0,...,h_K \in \mathbb{C}.
\end{equation}
The degree of ${\breve H}$ is defined as  ${\rm deg} \ {\breve H} = \max\{n \ | \ h_n\neq 0, n = 0, 1, ...,K \}$.
Furthermore, we define 
\begin{equation} \label{def:H_breve_conj}
    {\breve H}^*(z)=\sum_{n=0}^{K}h_n^*z^{-n}.
\end{equation}

To prove Lemma \ref{lemma:BRL}, we first introduce Lemma \ref{Thm:nonnegative_trigpoly} (Theorem 1.15 of \cite{Dumitrescu2017}), which characterizes a trigonometric polynomial that is nonnegative over an interval.

\begin{lemma} (Theorem 1.15 of \cite{Dumitrescu2017}) \label{Thm:nonnegative_trigpoly}
    Let \mbox{${\breve R}\in\mathbb{C}[z]$} with \mbox{${\rm deg}\ {\breve R} = K$}.
    The following statement holds: 
    \mbox{$R(f) \geq 0$}, \mbox{$\forall f \in [\alpha, \beta] \subset \left[-\tfrac{1}{2}, \tfrac{1}{2}\right]$}, if and only if ${\breve R}(z)$ can be expressed as:
    \begin{equation} \label{eq:relation_between_R_F_G}
        {\breve R}(z) = {\breve F}(z){\breve F}^*(z^{-1}) + E_{\alpha \beta}(z){\breve G}(z){\breve G}^*(z^{-1}),
    \end{equation}
    where ${\breve F}(z)$ and ${\breve G}(z)$ are causal polynomials defined in (\ref{eq:H})
    with \mbox{${\rm deg} \ {\breve F}\leq K$} and \mbox{${\rm deg} \ {\breve G}\leq K-1$}
    (but at least one of ${\rm deg} \ {\breve F}= n$ or  ${\rm deg} \ {\breve G}= n-1$ must hold), 
    ${\breve F}^*(z)$ and ${\breve G}^*(z)$ are defined by \eqref{def:H_breve_conj},
    \begin{align}
        &E_{\alpha \beta}(z) = d_0 + d_1^*z + d_1z^{-1},\\
        d_0 &= -\left( \tfrac{\tan\tfrac{\alpha}{2}\tan\tfrac{\beta}{2}+1}{2} \right),\label{def:d0} \\ 
        d_1 &= \tfrac{1-\tan\tfrac{\alpha}{2}\tan\tfrac{\beta}{2}}{4} + j\tfrac{\tan\tfrac{\alpha}{2}+\tan\tfrac{\beta}{2}}{4}. \label{def:d1}
    \end{align}
\end{lemma}
\begin{proof}
    The proof of Lemma \ref{Thm:nonnegative_trigpoly} is provided in Theorem 1.15 of \cite{Dumitrescu2017}.
\end{proof}

As expressing the polynomial in terms of its coefficients is preferred for formulating the optimization problem, the coefficients of the trigonometric polynomial ${\breve R}(z)$ defined in (\ref{eq:relation_between_R_F_G}) can be obtained through the following Lemma \ref{Thm:nonnegative_trace_parametrization}.
\begin{lemma} (A special case of Theorem 2 in \cite{Dumitrescu2006} and Corollary 4.16 in \cite{Dumitrescu2017}, considering a univariate trigonometric polynomial) \label{Thm:nonnegative_trace_parametrization}
    Let \mbox{${\breve R}\in\mathbb{C}[z]$} with \mbox{${\rm deg}\ {\breve R} = K$}.
    The following statement holds:     
    \mbox{$R(f) \geq 0$}, 
    \mbox{$\forall f \in [\alpha, \beta] \subset \left[-\tfrac{1}{2}, \tfrac{1}{2}\right]$}, if and only if there exists positive semidefinite matrices \mbox{${\bf Q} \in \mathbb{H}_+^{K+1}$} and \mbox{${\bf P} \in \mathbb{H}_+^K$} such that
    \begin{align} \label{eq:LMI}
        r_n 
        &= \mathrm{Tr}\left( \boldsymbol{\Theta}_{K+1}^{(n)}{\bf Q} \right) 
        + \mathrm{Tr}\Big( {\bf{\Phi}}_K^{(n)} {\bf P} \Big), \ \forall n \in \mathbb{Z}_{K+1},
    \end{align}
    where
    \begin{equation} \label{eq:Phi_K}
        {\bf{\Phi}}_K^{(n)} =d_0 \boldsymbol{\Theta}_{K}^{(n)} + d_1^* \boldsymbol{\Theta}_{K}^{(n+1)} 
        + d_1 \boldsymbol{\Theta}_{K}^{(n-1)},
    \end{equation}
    $d_0$ is defined in \eqref{def:d0}, $d_1$ is defined in \eqref{def:d1}, and ${\bf {\Theta}}_{K}^{(n)}$ is defined in (\ref{def:Elementary_Toeplitz_matrix}). 
    In addition, $r_{-n}=r_n^*$, $\forall n \in \mathbb{Z}_{K+1}$, holds.
    
    Note that Theorem 2 in \cite{Dumitrescu2006} and Corollary 4.16 in \cite{Dumitrescu2017} deal with a multivariate trigonometric polynomial, while Lemma \ref{Thm:nonnegative_trace_parametrization} specifically considers a univariate trigonometric polynomial \mbox{${\breve R}\in\mathbb{C}[z]$}.
\end{lemma}
\begin{proof}
    The proof of Lemma \ref{Thm:nonnegative_trace_parametrization}, provided in Appendix \ref{appendix:proof_thm_nonnegative_trace_parametrization}, is derived using our notation for clarity.
\end{proof}

Lastly, with Lemma \ref{Thm:nonnegative_trace_parametrization}, the proof of Lemma \ref{lemma:BRL} is shown below.
\begin{proof}
According to the Riesz-Fejér Theorem 
    \footnote{The Riesz-Fejér Theorem states that, given \mbox{${\breve R}\in\mathbb{C}[z]$} with \mbox{${\rm deg}\ {\breve R} = K$},
    the following holds:
    \mbox{${R}(f) \geq 0$}, \mbox{$\forall f \in \left[-\tfrac{1}{2}, \tfrac{1}{2}\right]$} if and only if there exists
    ${\breve H}(z)$, a causal polynomial defined in (\ref{eq:H}), such that ${\breve R}(z) = {\breve H}(z){\breve H}^*(z^{-1})$. 
    Additionally, on the unit circle, $R(f) = |H(f)|^2$. 
    For the proof, please refer to Theorem 1.1 in \cite{Dumitrescu2017}.},
given a causal polynomial $H(f)$ defined in (\ref{eq:causal_trig}),
we can have \mbox{${\breve R}\in\mathbb{C}[z]$} and
\begin{equation} \label{eq:R(f)}
    R(f) 
    = |H(f)|^2,
\end{equation}
where \mbox{$R(f) = \sum_{n=-K}^{K}r_ne^{-j2\pi fn}$} and $R(f)\geq 0$, $\forall f\in[\frac{-1}{2}, \frac{1}{2}]$.
Additionally, with (\ref{eq:R(f)}), $|H(f)|\leq \gamma$ in (\ref{BRL_inequality}) is equivalent to 
\begin{equation} \label{eq:R<=|H|^2}
    R(f)
    \leq \gamma^2.
\end{equation}
Let \mbox{$B(f)=\gamma^2=\sum_{n=-K}^{K}\gamma^2\delta_n e^{-j2 \pi fn}$},
where $\delta_n$ is defined in (\ref{def:delta}). 
Then (\ref{eq:R<=|H|^2}) can be written as $R(f)\leq B(f)$, leading to:
\begin{equation}
    \sum_{n=-K}^{K}(\gamma^2\delta_n-r_n)e^{-j2 \pi fn}\geq0, \ \forall f \in \left[-\tfrac{1}{2}, \tfrac{1}{2}\right].
\end{equation}  
Let $a_n=\gamma^2\delta_n-r_n$.
According to Lemma \ref{Thm:nonnegative_trace_parametrization}, there exist
${\bf K} \in \mathbb{H}_+^{K+1}$, ${\bf L} \in \mathbb{H}_+^{K}$,
${\bf Y} \in \mathbb{H}_+^{K+1}$, and ${\bf Z} \in \mathbb{H}_+^{K}$ such that
\begin{align}
    r_n &= \mathrm{Tr}\left( \boldsymbol{\Theta}_{K+1}^{(n)}{\bf K} \right) 
    + \mathrm{Tr}\Big( {\bf{\Phi}}_{K}^{(n)} {\bf L} \Big), \label{eq:r_n} \\
    a_n &= \mathrm{Tr}\left( \boldsymbol{\Theta}_{K+1}^{(n)}{\bf Y} \right) 
    + \mathrm{Tr}\Big( {\bf{\Phi}}_{K}^{(n)} {\bf Z} \Big).  \label{eq:a_n}
\end{align}
We can then derive $\gamma^2\delta_n = r_n + a_n =\mathrm{Tr}( \boldsymbol{\Theta}_{K+1}^{(n)}{\bf Q}) 
+ \mathrm{Tr}( {\bf{\Phi}}_{K}^{(n)} {\bf P} )$,
where we let \mbox{${\bf Q}={\bf K}+{\bf Y}$}, \mbox{${\bf P}={\bf L}+{\bf Z}$}, thus establishing (\ref{eq:BRL_LMI_a}).
Additionally, it follows that 
\begin{equation} \label{ineq:Q_K_P_L}
    {\bf Q}\succeq{\bf K} \text{ and } {\bf P}\succeq{\bf L}.
\end{equation}

Define \mbox{${\bf h}=[h_0, h_1,...,h_K]^T$} and $h_n,\ \forall n\in\{0,1,...,K\}$, are the coefficients of the causal polynomial $H(f)$ defined in (\ref{eq:causal_trig}).
Additionally, let \mbox{${\bf v}_K(f)=[1,e^{j2\pi f},...,e^{j2\pi f K}]^T$}.
We then have {$H(f)=\sum_{n = 0}^{K} h_ne^{-j2\pi fn}={\bf v}_K(f)^H{\bf h}$}.
Further, letting \mbox{${\bf H}={\bf h}{\bf h}^H\in \mathbb{H}_+^{K+1}$},
we obtain
\begin{small}
\begin{align}
    R(f) 
    &=|H(f)|^2 
    ={\bf v}_K(f)^H{\bf H}{\bf v}_K(f) 
    ={\rm Tr}\big({\bf v}_K(f){\bf v}_K(f)^H{\bf H}\big)  \nonumber\\
    &\hspace{-10mm}= {\rm Tr}\left(\sum_{n=-K}^{K}{\bf \Theta}^{(n)}_{K+1}e^{-j2\pi fn}{\bf H}\right)  
    = \sum_{n=-K}^{K} {\rm Tr}\big({\bf{\Theta}}^{(n)}_{K+1}{\bf H}\big) e^{-j2\pi fn}.  
\end{align}   
\end{small}

\vspace{-2mm}
\noindent
Then we have the coefficients of $R(f)$ be expressed as
\mbox{$r_n = {\rm Tr}({\bf{\Theta}}^{(n)}_{K+1}{\bf H})$}.
Furthermore, by comparing it  with (\ref{eq:r_n}),
we set ${\bf K}={\bf H}$ and ${\bf L}= {\bf O}_K$. 
As a result, we obtain ${\bf Q}\succeq{\bf H}={\bf h}{\bf h}^H$ and ${\bf P}\succeq{\bf O}_K$ according to (\ref{ineq:Q_K_P_L}).
By Schur complement, ${\bf Q}\succeq{\bf h}{\bf h}^H$ is equivalent to \eqref{eq:BRL_LMI_b}, and then the proof of Lemma \ref{lemma:BRL} is completed.

\end{proof}

\section{Proof of Lemma \ref{Thm:nonnegative_trace_parametrization}} \label{appendix:proof_thm_nonnegative_trace_parametrization}

Let \mbox{${\breve R}\in\mathbb{C}[z]$} with \mbox{${\rm deg}\ {\breve R} = K$}. 
According to Lemma \ref{Thm:nonnegative_trigpoly},
${\breve R}(z)$ can be expressed in the form of (\ref{eq:relation_between_R_F_G}).
Let $\boldsymbol{\psi}_K(z) = [1,z,...,z^K]^T$ and ${\bf f}=[f_0, f_1..,f_K]^T$,
and then we have
${\breve F}(z)=\sum_{n=0}^{K}f_nz^{-n}=\boldsymbol{\psi}_K^T(z^{-1}){\bf f}$ and
{${\breve F}^*(z^{-1})=\sum_{n=0}^{K}f_n^*z^{n}={\bf f}^H\boldsymbol{\psi}_K(z)$}.
Also, letting ${\bf Q}={\bf f} {\bf f}^H  \in \mathbb{H}_+^{K+1}$, we obtain 
\begin{footnotesize}
\begin{align}\label{eq:F(z)}
    {\breve F}(z){\breve F}^*(z^{-1})
    &=\boldsymbol{\psi}_K(z^{-1})^T{\bf Q}\boldsymbol{\psi}_K(z) 
    ={\rm Tr}\Big(\boldsymbol{\psi}_K(z)\boldsymbol{\psi}_K(z^{-1})^T{\bf Q}\Big) \nonumber\\
    &\hspace{-15mm}= {\rm Tr}\left(\sum_{n=-K}^{K}{\bf \Theta}^{(n)}_{K+1}z^{-n}{\bf Q}\right) 
    = \sum_{n=-K}^{K} {\rm Tr}\big({\bf{\Theta}}^{(n)}_{K+1}{\bf Q}\big) z^{-n}.
\end{align}        
\end{footnotesize}

\vspace{-2mm}
\noindent
Similarly, let {$\boldsymbol{\psi}_{K-1}(z)= [1,z,...,z^{K-1}]^T$}, ${\bf g}=[g_0, g_1, ...,g_{K-1}]^T$, and ${\bf P}={\bf g} {\bf g}^H \in \mathbb{H}_+^{K}$,
we get
${\breve G}(z){\breve G}^*(z^{-1}) = \sum_{n=-(K-1)}^{K-1} {\rm Tr}\big({\bf{\Theta}}^{(n)}_{K}{\bf P}\big)z^{-n}$.
Then, we can have (\ref{eq:relation_between_R_F_G}) rewritten as
\begin{footnotesize}
\begin{align} \label{eq:R_Q_P}
    \hspace{-3mm}{\breve R}(z) 
    &=\sum_{n=-K}^{K} {\rm Tr}\big({\bf{\Theta}}^{(n)}_{K+1}{\bf Q}\big)z^{-n} + d_0 \sum_{n=-(K-1)}^{K-1} {\rm Tr}\big({\bf{\Theta}}^{(n)}_{K}{\bf P}\big)z^{-n} \nonumber \\
    &+ \Big(d_1^*z + d_1z^{-1}\Big) \sum_{m=-(K-1)}^{K-1} {\rm Tr}\big({\bf{\Theta}}^{(m)}_{K}{\bf P}\big)z^{-m}.
\end{align}
\end{footnotesize}

\vspace{-2mm}
\noindent
Let ${\breve U}(z) =(d_1^*z + d_1z^{-1}) \sum_{m=-(K-1)}^{K-1} {\rm Tr}({\bf{\Theta}}^{(m)}_{K}{\bf P})z^{-m}$.        
Moreover, denote $p_m={\rm Tr}({\bf{\Theta}}^{(m)}_{K}{\bf P}), \ \forall m \in \{ -(K-1), ..., K-1 \}$, then we can derive:
\begin{footnotesize}
\begin{align}  
    \hspace{-2mm}
    {\breve U}(z) 
    &=\ d_1^*\sum_{m = -(K-1)}^{K-1} p_mz^{-(m-1)} + d_1\sum_{m = -(K-1)}^{K-1} p_mz^{-(m+1)}\nonumber\\
    &=\ d_1^*\sum_{n = -K}^{K-2} p_{n+1}z^{-n} + d_1\sum_{n = -K+2}^{K} p_{n-1}z^{-n} \nonumber\\
    &= \sum_{n = -K}^{-K+1} d_1^*p_{n+1}z^{-n} + \sum_{n = K-1}^{K} d_1p_{n-1}z^{-n} \nonumber \\
    &+ \sum_{n = -K+2}^{K-2} (d_1^*p_{n+1} + d_1p_{n-1})z^{-n}. \label{eq:U_brave} 
\end{align}        
\end{footnotesize}

\noindent
We can then rewrite \eqref{eq:U_brave} as ${\breve U}(z)=\sum_{n = -K}^{K} u_nz^{-n}$, where
\begin{small}
\begin{subequations}\label{eq:u_m}
\begin{align} 
    \hspace{-25mm} 
    u_n &= 
    \begin{cases}
        d_1^*p_{n+1} + d_1p_{n-1},&\ n \in \{0,1, \ldots, K-2\}\\
        d_1p_{n-1},&\ n \in \{ K-1, K\}
    \end{cases}  \label{eq:un_1}  \\  
    &= \mathrm{Tr}\left[ \left( d_1^*\boldsymbol{\Theta}_{K}^{(n+1)} + d_1\boldsymbol{\Theta}_{K}^{(n-1)} \right){\bf P} \right], \ n \in \{0, 1,  \ldots, K\}, \label{eq:un_2}
\end{align}
\end{subequations}
\end{small}

\noindent
and $u_{-n}=u_n^*$.
Notice that (\ref{eq:un_1}) is equal to (\ref{eq:un_2}) since ${\bf {\Theta}}_{K}^{(n)}={\bf O}_{K}$ if $|n|\geq K$ as specified in (\ref{def:Elementary_Toeplitz_matrix}).
We then can have (\ref{eq:R_Q_P}) be expressed as
\mbox{${\breve R}(z)=\sum_{n=-K}^{K}r_nz^{-n}$} with $r_n$, \mbox{$\forall n \in \{0,1,...,K\}$}, expressed as
\begin{small}
\begin{align} \label{eq:r_n_trace_param}
    r_n 
    &= \mathrm{Tr}\left( \boldsymbol{\Theta}_{K+1}^{(n)}{\bf Q} \right) 
    + \mathrm{Tr}\Big( \Big( d_0 \boldsymbol{\Theta}_{K}^{(n)} + d_1^* \boldsymbol{\Theta}_{K}^{(n+1)} 
    + d_1 \boldsymbol{\Theta}_{K}^{(n-1)} \Big) {\bf P} \Big) \nonumber \\ 
    &= \mathrm{Tr}\left( \boldsymbol{\Theta}_{K+1}^{(n)}{\bf Q} \right) +  \mathrm{Tr}\Big({\bf{\Phi}}_K^{(n)} {\bf P} \Big)
\end{align}
\end{small}

\vspace{-2mm}
\noindent
with ${\bf{\Phi}}_K^{(n)}$ defined in (\ref{eq:Phi_K})
and $r_{-n}=r_n^*$,   
we then complete the proof.

\bibliographystyle{IEEEtran}
\bibliography{IEEEabrv,Bibliography}

\vfill

\end{document}